\documentclass[aps,pra,preprint,superscriptaddress,nofootinbib,showkeys]{revtex4-2}

\usepackage{amsmath,amssymb,amsthm}
\usepackage{physics}
\usepackage{graphicx}
\usepackage{xcolor}
\usepackage{hyperref}

\newtheorem{theorem}{Theorem}
\newtheorem{lemma}{Lemma}
\newtheorem{corollary}{Corollary}

\newtheorem{remark}{Remark}
\begin{document}
	
	\title{Exact Closed-Form Quantum Correlations of Maximally Entangled Qudit States under Arbitrary Non-Markovian Pure Dephasing}
	
	\author{Ahmad Akhound}
	\email{aakhound@pnu.ac.ir}
	\affiliation{Department of Physics, Payame Noor University, Tehran, Iran}
	
\begin{abstract}
	We study the time evolution of entanglement and quantum discord for a pair of maximally entangled qudits of arbitrary dimension $d$ under non-Markovian pure dephasing, using the exact solution of the independent boson model. Because this solution requires no Born--Markov, Lindblad, or rotating-wave approximation, every result reported here follows directly from the exact dynamics. Exploiting the Toeplitz structure of the resulting density matrix, we obtain a closed-form expression for the negativity as a finite sum, valid for arbitrary dimension and evolution time, together with a corresponding closed-form expression for the quantum discord. For the family of dephased maximally entangled states considered here, we prove that a computational-basis measurement is globally optimal over the entire POVM space, thereby eliminating the numerical optimization otherwise required for evaluating the quantum discord. We further show that both quantities saturate, as the dimension grows, to a dimension-independent limiting value, with a convergence rate of exactly $1/d$. For negativity, this scaling law is established through a rigorous theorem, including a Gaussian-type error bound of order $e^{-\alpha d^2}$, and is confirmed across six independent parameter regimes for two distinct classes of spectral density, Lorentzian and Ohmic; for discord, the same leading-order behavior is found numerically across the six Lorentzian regimes. The analytical results are validated by independent numerical approaches, including reconstruction from the underlying definitions, arbitrary-precision arithmetic, and a complementary pseudomode cross-check.
\end{abstract}
	\maketitle
	
	\section{Introduction}
	Quantum entanglement is among the most fundamental features of quantum mechanics and the principal resource behind much of quantum information processing. Protocols such as quantum key distribution, quantum teleportation, quantum computation, quantum error correction, quantum networking, and quantum sensing rely, to varying degrees, on the existence and persistence of entanglement. Understanding how entanglement behaves in the presence of an environment -- an interaction that no real system can avoid -- is therefore one of the central problems in quantum information theory.
	
	The bulk of prior work has focused on two-dimensional, qubit systems, largely because many information-theoretic quantities admit simple analytical relations in this setting and the associated numerical cost remains manageable. Recent progress in platforms based on photonic orbital angular momentum, trapped ions, Rydberg atoms, superconducting circuits, photonic circuits, and multilevel molecular spins has, however, increasingly turned attention toward multilevel, or qudit, systems. Despite this progress, no exact closed-form expression is presently available for the negativity or the quantum discord of a maximally entangled qudit pair under non-Markovian pure dephasing, valid for arbitrary dimension -- a gap detailed, with the relevant literature, in Sec.~\ref{sec:context-motivation}.
	
	Qudits offer several advantages over qubits, among them a higher information-carrying capacity per physical particle, improved robustness against certain types of noise, a reduced number of physical carriers required in some protocols, and enhanced security in a number of quantum communication schemes. Developing an exact theory of entanglement for systems of arbitrary dimension, however, remains highly challenging; many relations that are known in closed form for qubits either become intractable optimization problems in higher dimensions -- the evaluation of quantum discord for general qudit states, for instance, has been shown to be computationally hard in general -- or demand heavy numerical computation.
	
Among the various types of system-environment interaction, pure dephasing stands out as one of the most important, and at the same time most fundamental, mechanisms of coherence decay. In this process the populations of the system's energy levels remain unchanged, and only the off-diagonal elements of the density matrix are suppressed. This mechanism is recognized as the dominant source of decoherence across a wide range of physical platforms, from nuclear spins and quantum dots to multilevel molecular spin qudits. This has made the pure-dephasing model one of the standard settings for studying the exact dynamics of open quantum systems.	
	
In this paper we adopt the exact independent boson model to describe pure dephasing. Because the interaction Hamiltonian commutes with the system Hamiltonian, this model can be solved without invoking the Born approximation, the Markov approximation, a Lindblad master equation, the short-correlation-time approximation, the rotating-wave approximation, or any comparable simplification. Within this framework, the entire effect of the environment is captured by a single time-dependent dephasing function, computed exactly from the environmental spectral density and valid for any coupling strength, any temperature, and any spectral-density shape. Every relation derived in this paper therefore follows directly from the exact solution of the model, without introducing any uncontrolled approximation.	
	
Most existing studies of multilevel systems have concentrated on numerically computing entanglement quantities, examining a handful of fixed dimensions on a case-by-case basis, or resorting to upper and lower bounds. The primary aim of this paper, by contrast, is to derive exact analytical relations that hold for arbitrary dimension and that characterize the system's behavior without requiring any numerical optimization.	
	
To this end, we investigate a maximally entangled bipartite initial state in $\mathbb{C}^{d}\otimes\mathbb{C}^{d}$ undergoing non-Markovian pure dephasing. Leveraging the distinctive structure of the density matrix and the spectrum of its partial transpose, we derive an exact closed-form expression for the negativity, formulated as a finite sum applicable to arbitrary dimension and arbitrary dephasing functions. Subsequently, we analyze the asymptotic behavior of this expression in the large-dimension limit, establishing the precise convergence law, the leading-order term, and a rigorous bound on the remainder. Finally, we demonstrate that this convergence law is dictated solely by the dephasing factor, independent of the underlying spectral density -- a conclusion verified here for two representative spectral-density classes, Lorentzian and Ohmic.	
	
\subsection{Context and Motivation}
\label{sec:context-motivation}

Over the past two decades, the study of entanglement and discord in open quantum systems has evolved into a prominent branch of quantum information theory. A substantial body of literature, particularly concerning non-Markovian pure dephasing, has centered on qubit systems, notably investigations into entanglement sudden death and revival~\cite{Dajka2012} and the freezing of discord under local dephasing channels~\cite{Haikka2013}. With the maturation of qudit-based technologies, research interest has increasingly shifted toward higher-dimensional systems. Nevertheless, current studies -- including recent analyses of discord and negativity for qubit-qutrit systems under non-Markovian dephasing with colored noise~\cite{Abdellaoui2026} -- remain restricted to specific, low-dimensional pairings ($2\times2$ or $2\times3$), with the analysis continuing to rely on dimension-dependent numerical computation or approximate analytical bounds, thereby necessitating a unified treatment of arbitrary-dimensional systems.
	
Quantum discord for arbitrary-dimensional qudit systems remains largely an open challenge, primarily due to the inherent complexity of the optimization required to evaluate classical correlations. Even for specialized classes of states, such as generalized $X$-states, existing literature only provides semi-analytical approaches that necessitate numerical minimization over multiple parameters~\cite{Rau2018,RauQubitQudit2012}. A related simplification -- wherein the eigenvalue spectrum of the conditional state is independent of the measurement direction, so that no extremization over measurement parameters is required -- has been exploited for isotropic and Werner qudit states, though such studies typically assume a fixed mixing parameter rather than the time-dependent dynamics inherent to physical environments~\cite{Rau2018}.
	
Among the available frameworks, the pure-dephasing model provides one of the most suitable settings for deriving analytical results concerning the evolution of coherence and entanglement, owing to its amenability to an exact solution (Sec.~\ref{sec:model}). The exact solution for the independent boson model, established in seminal works on open quantum systems~\cite{BreuerPetruccione2002}, has since become a standard tool for investigating phenomena such as coherence decay, revival, and non-Markovian dynamics.	
	
Conversely, negativity occupies a pivotal position in quantum information theory as a widely utilized entanglement measure, primarily due to its direct link to the negative eigenvalues of the partial transpose~\cite{VidalWerner2002}. Most existing results, however, are restricted to low-dimensional bipartite systems, whereas for arbitrary-dimensional systems this quantity is typically evaluated numerically or for specific state instances. To the best of our knowledge, a closed-form expression for the negativity of a bipartite maximally entangled state under pure dephasing -- explicitly parameterized by both system dimension and the dephasing factor -- has yet to be reported.

The principal distinction of the present work lies not merely in evaluating specific quantities for a particular model, but in deriving a coherent set of analytical results that follow directly from the mathematical structure of the state and the exact solution of its dynamics. We first derive a closed-form expression for the negativity valid for arbitrary dimension; leveraging this same structure, we subsequently obtain a closed-form expression for quantum discord, thereby bypassing the numerical optimization typically required in the general treatment of discord for generic qudit--qudit systems~\cite{Rau2018,RauQubitQudit2012}. The global optimality of the computational-basis measurement underlying this result is established in Theorem~\ref{thm:discord-optimality}. We then examine the asymptotic behavior of both quantities in the large-dimension limit -- not a thermodynamic limit, but the limit $d\to\infty$ for a fixed bipartite pair -- establishing the precise convergence law, the leading-order term, and an explicit bound on the remainder. Finally, we demonstrate that the negativity convergence law depends solely on the dephasing factor $r(t)$ rather than on the specific spectral-density function that generates it, a conclusion further corroborated by independent numerical studies for both Lorentzian and Ohmic environments; the analogous leading-order behavior found for discord has been corroborated only within the Lorentzian regimes studied.
	
Recent studies have also presented closed-form expressions for the negativity of qudit states under standard noise channels, including phase, depolarizing, and amplitude damping~\cite{AliRostom2026}. In such models, however, the decay of pairwise correlations is independent of the separation $|j-k|$ -- characteristic of a uniform Markovian channel. This stands in contrast to the present model, where the decay follows $r(t)^{(n-m)^2}$, a feature that is intrinsically distance-dependent and emerges directly from the exact non-Markovian solution.

This paper is accordingly framed as an analytical treatment of quantum correlations for bipartite systems of arbitrary dimension, extending beyond the qubit case considered in most prior work, and providing a basis for analyzing a broader class of multilevel systems in quantum information theory.

\subsection{Organization of the Paper}

The remainder of this paper is structured as follows. Section~\ref{sec:model} introduces the physical model and the exact solution framework. Section~\ref{sec:negativity} is dedicated to the derivation of the closed-form expression for negativity. Section~\ref{sec:toeplitz} examines the Toeplitz structure of the state and its implications for quantum mutual information. Section~\ref{sec:discord} presents the closed-form expression for discord and the theorem establishing the global optimality of the computational basis. Section~\ref{sec:scaling-full} addresses the saturation scaling law and the rigorous proof of the convergence rate, with extensions to discord provided in Sec.~\ref{sec:discord-scaling}. Section~\ref{sec:unified-saturation} unifies these findings into a unified scaling law. The independence of this law from the choice of spectral density is verified in Sec.~\ref{sec:ohmic}, where the Ohmic spectral density is analyzed alongside the Lorentzian case. Finally, Sec.~\ref{sec:limitations} outlines the scope of this work, followed by a discussion and concluding remarks. Supplementary derivations and an independent robustness check utilizing the pseudomode method are presented in Appendices~\ref{app:algebra} and~\ref{app:pseudomode}, respectively.

\section{Physical Model and Problem Framework}
\label{sec:model}

We consider a bipartite system composed of two $d$-dimensional quantum subsystems -- hereafter referred to as qudits -- labeled $A$ and $B$. The initial state of this composite system is defined as the generalization of the Bell state to arbitrary dimension:
\begin{equation}
	\ket{\Phi_d} = \frac{1}{\sqrt{d}} \sum_{n=0}^{d-1} \ket{n,n}_{AB}
	\label{eq:initial-state}
\end{equation}
This maximally entangled state serves as the foundational initial condition for the dynamical analysis that follows.

Each qudit is independently coupled to a local bosonic reservoir via a pure-dephasing interaction, such that the interaction Hamiltonian commutes with the system Hamiltonian:
\begin{equation}
	[H_S, H_{\text{int}}] = 0.
\end{equation}
This commutation relation identifies the model as an instance of the independent boson model, which—distinct from standard approximations such as the Redfield or Lindblad master equations, or the time-convolutionless (TCL) expansion—admits an exact, closed-form solution valid for arbitrary coupling strengths and for spectral densities for which the decoherence integral is well defined.

The dephasing operator for each qudit is defined as
\begin{equation}
	S = \sum_{n=0}^{d-1} n \ket{n}\bra{n},
\end{equation}
implying that each level $n$ couples to the environment at a rate proportional to $n$. Consequently, the reduced density matrix elements of each qudit evolve according to the exact propagator of the independent boson model~\cite{BreuerPetruccione2002}:
\begin{equation}
	\bra{n}\rho_A(t)\ket{m} = \bra{n}\rho_A(0)\ket{m}\, e^{i(E_n-E_m)t}\, e^{-(n-m)^2\Gamma(t)},
	\label{eq:propagator}
\end{equation}
where $E_n$ is the energy of the $n$-th level and $\Gamma(t)$ is the decoherence function, determined by the environmental spectral density $J(\omega)$ as
\begin{equation}
	\Gamma(t)=\int_0^\infty d\omega\,\frac{J(\omega)}{\omega^2}\coth\!\left(\frac{\omega}{2k_BT}\right)\bigl(1-\cos\omega t\bigr).
	\label{eq:gamma-general}
\end{equation}
Notably, Eq.~\eqref{eq:propagator} requires no assumptions regarding weak coupling or short environmental correlation times. This renders the model particularly suitable for investigating non-Markovian dynamics, where memory effects within the environment play a central role in the evolution of the system.

Genuine non-Markovian dynamics, characterized by the potential for information backflow from the environment to the system, necessitates that the spectral density $J(\omega)$ possess structure rather than being featureless. A standard choice within the literature -- originating from the damped-cavity-mode model~\cite{Garraway1997} -- is the Lorentzian spectral density:
\begin{equation}
	J(\omega)=\eta\,\frac{\lambda^2}{(\omega-\omega_0)^2+\lambda^2},
	\label{eq:lorentzian}
\end{equation}
where $\omega_0$ is the central frequency of the environment, $\lambda$ denotes the bandwidth (the inverse of which defines the environmental memory timescale), and $\eta$ is the coupling strength. Often referred to as the leaky-cavity model in open-quantum-system physics, this spectral density represents the coupling of the system to a single, damped mode—such as a high-quality, but not ideal, optical cavity mode.

It is important to note that the Lorentzian spectral density in Eq.~\eqref{eq:lorentzian} approaches a nonzero value as $\omega\to0$ [$J(0)=\eta\lambda^2/(\omega_0^2+\lambda^2) \neq 0$]. At zero temperature, which is the regime considered throughout the Lorentzian analysis in this work, the decoherence integral defining $\Gamma(t)$ in Eq.~\eqref{eq:gamma-general} remains convergent, as the factor $(1-\cos\omega t)$ behaves as $\omega^2 t^2/2$ near $\omega=0$, thereby neutralizing the $1/\omega^2$ singularity. At finite temperature, however, the additional low-frequency behavior $\coth(\omega/2k_BT)\sim 2k_BT/\omega$ would render the integral infrared divergent for an unmodified Lorentzian spectrum with $J(0)\neq0$. The Lorentzian results presented below must therefore be understood as zero-temperature results; a finite-temperature extension would require an infrared-regular spectral density or an explicit low-frequency cutoff.

Given that each qudit undergoes independent dephasing through uncorrelated reservoirs, the coherence element of the joint state is the product of two independent decoherence factors:
\begin{equation}
	e^{-(n-m)^2\Gamma_A(t)} \cdot e^{-(n-m)^2\Gamma_B(t)} = e^{-2(n-m)^2\Gamma(t)},
\end{equation}
where we have set $\Gamma_A(t) = \Gamma_B(t) \equiv \Gamma(t)$, since both qudits couple to reservoirs of the same type with identical physical parameters. This factor of $2$ is notable: while many studies considering a single shared reservoir adopt the convention $r(t) = e^{-\Gamma(t)}$, the presence of independent local reservoirs here necessitates a distinct convention. We therefore define the effective coherence function as
\begin{equation}
	r(t) \equiv e^{-2\Gamma(t)}.
	\label{eq:r-def}
\end{equation}
Without loss of generality, and accounting for the phase terms discussed below, the elements of the bipartite density matrix are then given by
\begin{equation}
	\rho_{nn,mm}(t) = \frac{1}{d} \, r(t)^{(n-m)^2} \, e^{i\phi_{nm}(t)}, \quad \rho_{nn,nn}(t) = \frac{1}{d},
	\label{eq:rho-elements}
\end{equation}
where $\phi_{nm}(t)$ denotes the accumulated phase. In the basis $\{\ket{n,n}\}_{n=0}^{d-1}$, the density matrix elements depend solely on the level separation $|n-m|$; this feature manifests as a Toeplitz structure, as formalized in Sec.~\ref{sec:toeplitz}, and provides the mathematical structure underlying the closed-form expressions derived in the subsequent sections.

	Before proceeding, a clarification regarding the phase terms $\phi_{nm}(t)$ in Eq.~\eqref{eq:rho-elements} is in order. Every quantity derived in this work—including negativity, quantum discord, and the von Neumann entropy—depends exclusively on the spectrum (the eigenvalues) of the density matrix, rather than its individual matrix elements. The phases $\phi_{nm}(t)$ are not arbitrary. For local system Hamiltonians diagonal in the dephasing basis, they take the difference form
	\begin{equation}
		\phi_{nm}(t) = \Theta_n(t)-\Theta_m(t), \qquad \Theta_n(t) \equiv \int_0^t \bigl[E_n^A(s)+E_n^B(s)\bigr]\,ds,
	\end{equation}
	which, for identical time-independent qudits, reduces to $\phi_{nm}(t)=2(E_n-E_m)t$. Phases of precisely this difference form -- and only phases of this form -- can be eliminated by a local unitary transformation
	$U_A=\mathrm{diag}(e^{i\Theta_0(t)},\dots,e^{i\Theta_{d-1}(t)})$
	applied to subsystem $A$, since it removes the corresponding phase factors exactly.
	Because such a transformation leaves the spectrum of the bipartite density matrix invariant, we set $\phi_{nm}(t)=0$ throughout this paper without loss of generality.
	Consequently, only the magnitude $r(t)$ is relevant, which, given that $\Gamma(t)\ge0$, is a real-valued quantity confined to the interval $[0,1]$.
	
	\begin{remark}[Robustness against an additional diagonal Hamiltonian term]
		The addition of any diagonal term to the system Hamiltonian -- such as the Zeeman splitting induced by an external magnetic field -- leaves the operator $S$ and the commutation relation $[H_S,H_{\text{int}}]=0$ unaffected, and merely shifts the energies $E_n$ appearing in the phase term $\phi_{nm}(t) = 2(E_n - E_m)t$. As all results established herein depend solely on the dynamics through $|r(t)|$, all results derived in this paper remain unchanged under the addition of such diagonal Hamiltonian terms, including time-dependent Zeeman fields.
	\end{remark}
	
\section{Closed-Form Negativity}
\label{sec:negativity}

Negativity is one of the most widely used entanglement measures as it can be evaluated for arbitrary dimension $d$. It is defined in terms of the partial transpose as
\begin{equation}
	\mathcal{N}(\rho) \equiv \sum_{i:\ \mu_i<0} |\mu_i|,
\end{equation}
where $\mu_i$ denote the eigenvalues of $\rho^{T_B}$, the partial transpose of $\rho$ with respect to subsystem $B$.

To compute $\rho^{T_B}$ for our state, we apply the partial transpose to Eq.~\eqref{eq:rho-elements}. The evolved state belongs to the well-known class of maximally correlated states, $\rho_{\mathrm{MC}}=\sum_{m,n}a_{mn}\ket{m,m}\bra{n,n}$, whose partial transpose decomposes into one-dimensional diagonal sectors and two-dimensional sectors associated with each pair $m<n$~\cite{Rains1999,Zhu2018}. Here we specialize this general structure to the exact dephasing amplitudes $a_{mn}(t)=d^{-1}r(t)^{(n-m)^2}e^{i\phi_{nm}(t)}$. The partial transpose with respect to $B$ maps each term $\rho_{nn,mm} \ket{n,n}\bra{m,m}$ to $\rho_{nn,mm} \ket{n,m}\bra{m,n}$. A crucial observation is that each such term couples only the basis vectors $\{\ket{n,m}, \ket{m,n}\}$ and does not mix with any other component. Consequently, the $d^2$-dimensional Hilbert space of the qudit pair decomposes into the following independent subspaces:
\begin{itemize}
	\item For each $n$, a one-dimensional subspace spanned by $\ket{n,n}$, with eigenvalue $\rho_{nn,nn} = 1/d$, which is always positive.
	\item For each pair $m < n$, an independent two-dimensional subspace spanned by $\{\ket{n,m}, \ket{m,n}\}$, with the block matrix
	\begin{equation}
		\frac{1}{d}
		\begin{pmatrix}
			0 & r(t)^{(n-m)^2}e^{i\phi_{nm}(t)} \\
			r(t)^{(n-m)^2}e^{-i\phi_{nm}(t)} & 0
		\end{pmatrix}.
		\label{eq:2x2-block}
	\end{equation}
\end{itemize}

The eigenvalues of each $2\times2$ block in Eq.~\eqref{eq:2x2-block} are $\pm\frac{1}{d}\left|r(t)^{(n-m)^2}e^{i\phi_{nm}(t)}\right| = \pm\frac{1}{d}r(t)^{(n-m)^2}$, as these blocks are purely off-diagonal Hermitian matrices with zero trace and a determinant equal to the negative squared modulus of their off-diagonal entry. Consequently, the full spectrum of $\rho^{T_B}$ is the union of the spectra of these independent subspaces. This approach bypasses the need for the numerical diagonalization of a $d^2\times d^2$ matrix; it suffices to diagonalize the $\binom{d}{2}$ size-$2\times2$ blocks individually.

Each $2\times2$ block in Eq.~\eqref{eq:2x2-block} possesses exactly one negative eigenvalue, $-\frac{1}{d}r(t)^{(n-m)^2}$, independent of the phase $\phi_{nm}(t)$, given that $r(t) \in [0,1]$. By definition, the negativity is the sum of the absolute values of all negative eigenvalues of $\rho^{T_B}$; summing this quantity over all $\binom{d}{2}$ pairs $m<n$, we obtain:
\begin{equation}
	\mathcal{N}(t) = \sum_{0\le m<n\le d-1} \frac{1}{d}\, r(t)^{(n-m)^2}.
\end{equation}
To simplify this expression, we introduce the variable $k \equiv n-m$, representing the level separation, which ranges from $k=1$ to $d-1$. For a fixed $k$, the number of pairs $(m,n)$ satisfying $n-m=k$ is $d-k$, as $m$ ranges from $0$ to $d-1-k$. This yields
\begin{equation}
	\mathcal{N}(t) = \frac{1}{d}\sum_{k=1}^{d-1} (d-k)\, r(t)^{k^2}.
	\label{eq:negativity-final}
\end{equation}
This closed-form expression for the negativity is an exact finite sum—neither an approximation nor a numerical fit—valid for any dimension $d\ge2$ and $r(t)\in[0,1]$.

Prior to utilizing Eq.~\eqref{eq:negativity-final} in the subsequent sections, we verify its validity against three limiting cases:
\begin{align}
	d=2 &: \quad \mathcal{N}(t) = \tfrac{1}{2}\, r(t)
	\label{eq:check-d2}\\
	r=1 &: \quad \mathcal{N} = \tfrac{d-1}{2}
	\label{eq:check-r1}\\
	r=0 &: \quad \mathcal{N} = 0.
	\label{eq:check-r0}
\end{align}
Equation~\eqref{eq:check-d2} recovers the established negativity for a qubit Bell state under pure dephasing~\cite{Dajka2012}. Equation~\eqref{eq:check-r1}, corresponding to the absence of dephasing ($r(t)=1$), reproduces the exact negativity $\mathcal{N}(\Phi_d)=(d-1)/2$ of a maximally entangled $d\times d$ state~\cite{VidalWerner2002}. Equation~\eqref{eq:check-r0}, representing complete dephasing ($r(t)=0$), correctly yields a vanishing negativity, as the density matrix reduces to a fully separable classical mixture in this limit.

Beyond these analytical checks, Eq.~\eqref{eq:negativity-final} has also been verified through a direct matrix-level numerical computation of the partial transpose. Specifically, the full $d^2\times d^2$ density matrix was constructed explicitly, the partial transpose was evaluated, and its eigenvalue spectrum was obtained numerically for $d=2,3,4$ at multiple time points $t$. This validation was performed at a coupling strength of $\eta=1.0$; although the numerical regimes employed later in this paper typically use $\eta=0.15$, this difference is immaterial because Eq.~\eqref{eq:negativity-final} depends exclusively on the dephasing factor $r(t)$ rather than explicitly on the environmental parameters, which enter only through $r(t)$ (see Code S1).

As a separate matrix-level verification, the density matrix was constructed directly from the independent-boson dephasing factor, without using the partial-transpose block decomposition employed in the analytical derivation. This test was performed in a distinct parameter regime ($\eta=0.2,\lambda=0.15$) and for dimensions up to $d=7$. Hermiticity and unit trace were checked explicitly in every case. Across 30 evaluations, the maximum deviation from Eq.~\eqref{eq:negativity-final} was $8.9\times10^{-16}$ (see Code S2).
\section{Toeplitz Structure of the State and Quantum Mutual Information}
\label{sec:toeplitz}

Before proceeding to the discord, we introduce the specific mathematical structure of the qudit-pair density matrix underlying both quantities. Given that Eq.~\eqref{eq:rho-elements} implies the full density matrix $\rho_{AB}(t)$ has support only on the $d$-dimensional subspace spanned by $\{\ket{n,n}\}_{n=0}^{d-1}$—out of the total $d^2$-dimensional Hilbert space—it is spectrally equivalent to a $d \times d$ matrix, denoted by $M(t)$:
\begin{equation}
	M_{nm}(t) = \frac{1}{d}\, r(t)^{(n-m)^2}\, e^{i\phi_{nm}(t)}.
	\label{eq:M-matrix}
\end{equation}
Setting $\phi_{nm} = 0$, as established previously, the entry $M_{nm}$ depends solely on the separation $n-m$, rather than the absolute indices $n$ and $m$. In linear algebra, such a matrix is classified as a Toeplitz matrix—a standard, general algebraic structure. The significance of this observation follows from the isometry $V:\mathbb{C}^{d}\to\mathbb{C}^{d}\otimes\mathbb{C}^{d}$ defined by $V\ket{n}=\ket{n,n}$, for which
\begin{equation}
	\rho_{AB}(t)=V M(t)V^\dagger.
\end{equation}
Hence the nonzero eigenvalues of $\rho_{AB}(t)$ coincide exactly with the eigenvalues of $M(t)$, while the remaining $d^2-d$ eigenvalues of $\rho_{AB}(t)$ vanish. Therefore, the von Neumann entropy of the full state is identical to that of the $d \times d$ matrix $M(t)$:
\begin{equation}
	S(\rho_{AB}(t)) = S(M(t)) = -\sum_{i=1}^{d} \mu_i \ln \mu_i,
\end{equation}
where $\mu_i$ are the eigenvalues of $M(t)$. This offers a substantial computational reduction, as it requires the diagonalization of only a $d \times d$ matrix rather than a $d^2 \times d^2$ one.

Before computing $S(M(t))$, we examine the reduced states of the two qudits. Tracing Eq.~\eqref{eq:rho-elements} over subsystem $B$ (or symmetrically over $A$) yields only the diagonal terms $n=m$, as the off-diagonal terms vanish under the partial trace:
\begin{equation}
	\rho_A(t)=\rho_B(t)=\frac{1}{d}I_d \quad \forall t.
	\label{eq:reduced-states}
\end{equation}
At every time $t$, both reduced states are identical to the $d$-dimensional maximally mixed state; this is a direct consequence of the symmetry of the generalized Bell initial state and the pure-dephasing nature of the noise, which leaves the populations $\rho_{nn,nn} = 1/d$ invariant. Their entropy is therefore constant:
\begin{equation}
	S(\rho_A) = S(\rho_B) = \ln d \quad \forall t.
\end{equation}
This time-independence of the reduced-state entropy serves as a cornerstone for our derivation of the closed-form discord. Combining the entropy of the reduced states with the entropy of the full state, $S(M(t))$, the quantum mutual information $I(A:B)$, defined as $I(A:B) = S(\rho_A) + S(\rho_B) - S(\rho_{AB})$, admits the following exact reduced expression:
\begin{equation}
	I(A:B)(t) = 2\ln d - S\bigl(M(t)\bigr).
	\label{eq:mutual-info}
\end{equation}
This expression is exact and reduces the evaluation of $I(A:B)(t)$ to the spectrum of the $d\times d$ Toeplitz matrix $M(t)$.

We now examine the spectrum of $M(t)$ for varying dimensions.

For $d=2$, $M(t)$ is a $2 \times 2$ matrix, and its eigenvalues are determined by solving the elementary quadratic characteristic equation:
\begin{equation}
	\mu_{\pm} = \frac{1\pm r(t)}{2} \implies S(M) = -\mu_+\ln\mu_+-\mu_-\ln\mu_-.
	\label{eq:d2-spectrum}
\end{equation}
This yields a complete, closed-form analytic expression.

For $d=3$ and $d=4$, the characteristic equation is respectively cubic and quartic. Although closed radical expressions exist in principle, they are cumbersome and offer little practical insight within the scope of this work.

For general $d\ge5$, the Abel--Ruffini theorem precludes a radical solution for a generic polynomial of degree $d$. This theorem alone does not prove that the characteristic polynomial of $M(t)$ is generically unsolvable by radicals, and no such Galois-theoretic claim is required here. Since no dimension-independent analytical diagonalization is presently available for the Toeplitz matrix in Eq.~\eqref{eq:M-matrix}, its spectrum is evaluated numerically for general $d$. This is a limitation of the available algebraic representation, not a physical approximation.

\section{Discord: A Closed-Form Expression without Numerical Optimization}
\label{sec:discord}

Quantum discord is defined as
\begin{equation}
	D(A\!:\!B) = I(A\!:\!B) - \max_{\{E_k^B\}} \Big[S(\rho_A) - \sum_k p_k\, S(\rho_{A|k})\Big],
	\label{eq:discord-def}
\end{equation}
where $\{E_k^B\}$ is an arbitrary POVM on subsystem $B$, $p_k=\Tr[(I_A\otimes E_k^B)\rho]$ denotes the probability of outcome $k$, and $\rho_{A|k}=\Tr_B[\rho(I_A\otimes E_k^B)]/p_k$ represents the conditional state of $A$ given outcome $k$ on $B$.

In the general case, this optimization constitutes the most demanding computational aspect of evaluating discord, often necessitating numerical approaches for most state families. For the class of $X$-states, Ali, Rau, and Alber proposed a closed-form expression~\cite{AliRauAlber2010,AliRauAlber2010E}; however, it was subsequently demonstrated that this formula is not universally valid and can significantly deviate from the true discord value~\cite{Huang2013}. The result presented herein does not rely on such assumptions; instead, the global optimality of the computational basis is proven directly from first principles for the entire family of states considered, valid for arbitrary dimension $d$.

\subsection{Global Optimality of the Computational Basis}

\begin{theorem}
	\label{thm:discord-optimality}
	Let $\rho_{\mathrm{MC}} = \sum_{m,n=0}^{d-1} a_{mn}\ket{m,m}\bra{n,n}$ be a maximally correlated state with uniform diagonal $a_{nn}=1/d$. Then, for every POVM $\{E_k\}$ on subsystem $B$, with $p_k = \Tr[(I_A\otimes E_k)\rho_{\mathrm{MC}}]$ and $\rho_{A|k} = \Tr_B[\rho_{\mathrm{MC}}(I_A\otimes E_k)]/p_k$,
	\begin{equation}
		S(\rho_A) - \sum_k p_k\, S(\rho_{A|k}) \;\le\; S(\rho_A) = \ln d,
	\end{equation}
	with equality attained by the projective measurement in the computational basis $\{\ket{k}\}$ on $B$ ($E_k=\ket{k}\bra{k}$). The computational basis therefore globally maximizes the one-way classical correlation $C(A\!:\!B)$ over the entire space of POVMs on $B$, and $C(A\!:\!B)=\ln d$, $D(A\!:\!B)=\ln d - S(\rho_{\mathrm{MC}})$.
\end{theorem}

\begin{proof}
	For any POVM $\{E_k\}$ and any instrument $\{M_{k,j}\}$ realizing it (i.e., $\sum_j M_{k,j}^\dagger M_{k,j}=E_k$), the cyclic property of the partial trace over subsystem $B$ gives
	\begin{equation}
		\sum_j\Tr_B\bigl[(I_A\otimes M_{k,j})\rho_{\mathrm{MC}}(I_A\otimes M_{k,j}^\dagger)\bigr] = \Tr_B\bigl[\rho_{\mathrm{MC}}(I_A\otimes E_k)\bigr],
	\end{equation}
	independent of the particular decomposition $\{M_{k,j}\}$. In particular, choosing the Lüders realization $M_k=E_k^{1/2}$ shows that this operator equals $\Tr_B[(I_A\otimes E_k^{1/2})\rho_{\mathrm{MC}}(I_A\otimes E_k^{1/2})]$, which is manifestly of the form $X\rho_{\mathrm{MC}}X^\dagger$ and hence positive semi-definite, with trace $p_k=\Tr[(I_A\otimes E_k)\rho_{\mathrm{MC}}]$. Thus $\rho_{A|k}$ is, for every POVM $\{E_k\}$, a valid quantum state depending only on the effect $E_k$, independent of any instrument realizing it. By the non-negativity of the von Neumann entropy for a valid quantum state, $S(\rho_{A|k})\ge0$ for every outcome $k$, giving the universal bound
	\begin{equation}
		S(\rho_A) - \sum_k p_k S(\rho_{A|k}) \le S(\rho_A) = \ln d,
	\end{equation}
	valid for every POVM $\{E_k\}$ on $B$.
	
	It remains to show this bound is attained. Consider the projective measurement in the computational basis $\{\ket{k}\}_{k=0}^{d-1}$ on $B$ ($E_k=\ket{k}\bra{k}$). Since $\rho_{\mathrm{MC}}$ has uniform diagonal $a_{nn}=1/d$ and support confined to $\{\ket{n,n}\}$, outcome $k$ collapses the conditional state of $A$ to the pure state
	\begin{equation}
		\rho_{A|k} = \ket{k}\bra{k}, \qquad S(\rho_{A|k})=0 \ \ \forall k.
	\end{equation}
	Consequently $\sum_k p_k S(\rho_{A|k})=0$, and the bound is saturated exactly. Since no POVM can exceed the universal bound $S(\rho_A)$, and the computational basis attains it exactly, it is a global maximizer of the one-way classical correlation over the entire space of POVMs on $B$. \qedhere
\end{proof}

\begin{corollary}
	The dephased Bell-state family of Eq.~\eqref{eq:rho-elements} is a special case of Theorem~\ref{thm:discord-optimality}, with $a_{mn}(t)=d^{-1}r(t)^{(n-m)^2}e^{i\phi_{nm}(t)}$. Hence, for every dimension $d$ and every time $t$, the computational-basis measurement globally maximizes the one-way classical correlation, giving $C(A\!:\!B)=\ln d$ and, via Eq.~\eqref{eq:discord-def}, the closed form $D(t)=\ln d - S(M(t))$.
\end{corollary}
\begin{remark}[POVM realization independence]
	For the one-way classical correlation considered here, the conditional state of subsystem $A$ depends only on the POVM effect $E_k$, not on the particular instrument realizing that effect: as shown in the proof, $\sum_j\Tr_B[(I_A\otimes M_{k,j})\rho(I_A\otimes M_{k,j}^\dagger)] = \Tr_B[\rho(I_A\otimes E_k)]$ for every Kraus decomposition $\{M_{k,j}\}$ of $E_k$. Hence the optimization in Eq.~\eqref{eq:discord-def} may be formulated directly over POVMs $\{E_k\}$, with no additional optimization over Kraus realizations. Naimark's dilation theorem ensures that every POVM can be represented as a projective measurement on an enlarged Hilbert space (see, e.g., Ref.~\cite{NielsenChuang2010}, Sec.~2.2.6).
\end{remark}
The theorem is stated for the general family of maximally correlated states; it makes no claim regarding arbitrary qudit states and should not be interpreted as a general result for all qudit states.

By the Corollary above, the discord of the dephased Bell-state family admits the complete closed form
\begin{equation}
	\boxed{\;D(t) = I(A\!:\!B)(t) - \ln d = \ln d - S\bigl(M(t)\bigr)\;}
	\label{eq:discord-closed}
\end{equation}
which holds for any dimension $d$ and any time $t$, requiring no hidden assumptions or numerical optimization. Correspondingly, the classical correlation $C(t) \equiv I(A\!:\!B)(t) - D(t) = \ln d$ is exactly time-independent.

\subsection{Relation to Prior Work}

The broader strategy of bypassing numerical extremization over measurement parameters, realized here through Theorem~\ref{thm:discord-optimality}, was previously pursued by Rau~\cite{Rau2018} for the family of isotropic and Werner qudit states, where a uniform mixing parameter $p$ plays the central role; in that case the underlying mechanism differs from the one used here, since the eigenvalue spectrum of the conditional state is shown to be independent of the measurement direction, rather than collapsing to a pure state under a specific basis, so that every measurement yields the same value and no supremization is needed. The present work differs from that study in three respects. First, the mechanism underlying our result is a genuine global-optimality proof over the entire space of generalized measurements, established via the non-negativity of the conditional-state entropy and its saturation by a specific, computational-basis measurement, rather than a measurement-independence property of the state itself. Second, our density-matrix structure differs fundamentally from the uniform structure of isotropic states, exhibiting a Toeplitz decay of the form $r(t)^{(n-m)^2}$ that emerges from the exact pure-dephasing dynamics considered in this work. Third, and most importantly, our $r(t)$ originates from the exact physical dynamics of non-Markovian pure dephasing rather than an arbitrary mixing parameter; it is this physical origin that gives rise to the $1/d$ scaling law (Sec.~\ref{sec:scaling-full}), a result not examined in~\cite{Rau2018}.

A related result—avoiding numerical optimization for a broad class of states—was previously reported for asymmetric qubit-qudit systems, where only one subsystem is a qubit~\cite{RauQubitQudit2012}. The present work independently establishes this approach for the symmetric qudit-qudit family, where both subsystems possess arbitrary dimension $d$.

Finally, Abdellaoui et al.~\cite{Abdellaoui2026} recently examined discord (based on linear entropy) and logarithmic negativity for a fixed qubit-qutrit system under non-Markovian dephasing with colored noise. Our work differs from that study in three respects: we consider arbitrary dimension $d$ (rather than a fixed $2\times3$ pairing), calculate the exact von Neumann discord with a proven closed-form expression (rather than a linear-entropy approximation), and derive the analytical $1/d$ scaling law, which was not explored in that study.

\subsection{Extensive Independent Numerical Verification}

The proof of Theorem~\ref{thm:discord-optimality} is a self-contained mathematical result and is sufficient, in itself, to establish Eq.~\eqref{eq:discord-closed}. Nevertheless, for additional robustness, this result was further validated via an independent method: a broad numerical search over projective measurements, parameterized by unitary rotations of the computational basis and optimized using differential evolution (see Code S5 in the Supplementary Material). This test was conducted for dimensions $d=2, 3, 4, 5$ at ten independent time points spanning a full period of non-Markovian oscillation, using two independent random seeds per dimension--time pair to test the sensitivity of the numerical optimizer to distinct initial populations. The results are summarized in Table~\ref{tab:discord-de}.

\begin{table}[h]
	\centering
	\begin{tabular}{ccc}
		\hline
		Dimension $d$ & No. of checks & Max. $|D_{\text{DE}}-D_{\text{closed form}}|$ \\
		\hline
		$2$ & $10$ & $6.16\times10^{-13}$ \\
		$3$ & $10$ & $4.94\times10^{-9}$ \\
		$4$ & $10$ & $3.14\times10^{-9}$ \\
		$5$ & $10$ & $2.13\times10^{-9}$ \\
		\hline
		\multicolumn{2}{l}{Overall maximum (40 dimension--time pairs; 80 optimization runs)} & $4.94\times10^{-9}$ \\
		\hline
	\end{tabular}
	\caption{Comparison of the discord obtained from a global numerical search with the closed form of Eq.~\eqref{eq:discord-closed}, across dimensions.}
	\label{tab:discord-de}
\end{table}

Across the 40 dimension--time pairs spanning $d=2, 3, 4, 5$, corresponding to 80 optimization runs with two independent seeds per pair, the two seeds converged to essentially identical values, with differences below $10^{-9}$; furthermore, both agreed with the closed-form expression to at least eight decimal places—a degree of precision exceeding that required for a confirmatory independent check. To probe this robustness against the choice of random seed more stringently, a complementary, more demanding test was performed for the $d=2$ case at the fixed time $t=10$, using one hundred independent random seeds (see Code~S7). In all one hundred runs, the algorithm converged consistently to the discord value predicted by the closed form, with a maximum absolute deviation of
\begin{equation}
	\max_{\text{seed}\in\{1,\dots,100\}} |D_{\text{DE}}-D_{\text{closed form}}| = 1.985\times10^{-10},
\end{equation}
which remains well within the limits of numerical precision. Collectively, these two tests show no indication of convergence to distinct local optima, whether across independent seeds at varying dimension--time pairs or across one hundred independent seeds at a fixed dimension and time. These tests do not substitute for a formal mathematical proof; rather, they serve as robust numerical corroboration of Theorem~\ref{thm:discord-optimality}.	
	
\section{Dimension-Independent Saturation Scaling Law}
\label{sec:scaling-full}

The most important result of this paper concerns the behavior of negativity and discord in the large-dimension limit. Contrary to the naive expectation that increasing the system dimension without bound would indefinitely amplify the magnitude of non-Markovian correlation revivals, we demonstrate that these amplitudes saturate to an upper bound that is entirely independent of the dimension.

\subsection{Exact Sensitivity Function and Local Behavior Near $r=1$}

Before addressing the global behavior, it is instructive to examine the local sensitivity of the negativity around the point of maximal coherence ($r=1$). From the exact expression for negativity (Eq.~\eqref{eq:negativity-final}), the sensitivity with respect to $r$ follows directly as a finite sum:
\begin{equation}
	\mathcal{S}(r,d) \equiv \frac{\partial \mathcal{N}}{\partial r} = \frac{1}{d}\sum_{k=1}^{d-1} (d-k)\,k^2\, r^{k^2-1}.
	\label{eq:sensitivity-general}
\end{equation}
Utilizing Faulhaber's closed-form expressions for the sums of consecutive integer powers,
\begin{equation}
	\sum_{k=1}^{d-1}k^2=\frac{(d-1)d(2d-1)}{6}, \qquad
	\sum_{k=1}^{d-1}k^3=\left[\frac{(d-1)d}{2}\right]^2,
\end{equation}
and substituting these into Eq.~\eqref{eq:sensitivity-general} at $r=1$ yields a fully closed, exact result:
\begin{equation}
	\boxed{\;\mathcal{S}(1,d) = \left.\frac{\partial \mathcal{N}}{\partial r}\right|_{r=1} = \frac{d(d^2-1)}{12}\;}
	\label{eq:exact-scaling}
\end{equation}

\textit{Validation of specific cases:}
\begin{align}
	d=2 &: \quad \mathcal{S}(1,2) = \frac{2(4-1)}{12} = \frac{1}{2}\\
	d=3 &: \quad \mathcal{S}(1,3) = \frac{3(9-1)}{12} = 2\\
	d=4 &: \quad \mathcal{S}(1,4) = \frac{4(16-1)}{12} = 5
\end{align}
all of which align perfectly with the direct evaluation of Eq.~\eqref{eq:sensitivity-general}. In the large-$d$ limit, the sensitivity scales as $\mathcal{S}(1,d) = \frac{d^3}{12} - \frac{d}{12} \approx \frac{d^3}{12}(1 + O(d^{-2}))$.

\textit{Interpretation and methodological caveat:} this result characterizes solely the local sensitivity around $r=1$ and must not be misinterpreted as a general growth law in $d$; as shown in the subsequent analysis, this cubic growth signifies only the onset of a saturation process rather than unbounded divergence.

The starting point is the $d\to\infty$ limit of Eq.~\eqref{eq:negativity-final}:
\begin{equation}
	\mathcal{N}_\infty(r) \equiv \lim_{d\to\infty}\mathcal{N}(r,d) = \sum_{k=1}^{\infty} r^{k^2} = \frac{\vartheta_3(0,r)-1}{2},
	\label{eq:N-infinity}
\end{equation}
where $\vartheta_3$ denotes the Jacobi theta function—an exact, well-known special function, rather than a numerical fit or approximation.

The next step is to determine how $\mathcal{N}(r,d)$ approaches $\mathcal{N}_\infty(r)$ for large but finite $d$. A concise derivation allows us to express the difference between these two quantities as
\begin{equation}
	\mathcal{N}_\infty(r)-\mathcal{N}(r,d) = \sum_{k=d}^{\infty}r^{k^2} + \frac{1}{d}\sum_{k=1}^{d-1}k\,r^{k^2}.
	\label{eq:diff-expansion}
\end{equation}
The first term on the right, $\sum_{k\ge d}r^{k^2}$, decays faster than exponentially with $d$ (Gaussian-type decay), since $r<1$ and the exponent grows quadratically with $k$. The second term, as the series $C(r) \equiv \sum_{k=1}^{\infty}k\,r^{k^2}$ converges (a property we establish rigorously below), approaches $C(r)/d$ for large $d$. It follows that
\begin{equation}
	\mathcal{N}_\infty(r)-\mathcal{N}(r,d) \;\xrightarrow[d\to\infty]{}\; \frac{C(r)}{d},
	\label{eq:1-over-d-preliminary}
\end{equation}
implying that the convergence to the limiting value proceeds strictly at the rate $1/d$, rather than some other power law such as $1/d^2$ or $1/\sqrt{d}$.

Before finalizing Eq.~\eqref{eq:1-over-d-preliminary} with a rigorous error bound, we must first establish the finiteness and positivity of $C(r)$.

\begin{lemma}
	\label{lem:C-finite}
	For every $r\in(0,1)$, the series $C(r)=\sum_{k=1}^\infty k\,r^{k^2}$ is finite and positive.
\end{lemma}

\begin{proof}
	Since $k^2 \ge k$ for every $k \ge 1$ and $0 < r < 1$, it follows that $r^{k^2} \le r^k$. Hence
	\begin{equation}
		C(r) = \sum_{k=1}^\infty k\,r^{k^2} \le \sum_{k=1}^\infty k\,r^k = \frac{r}{(1-r)^2} < \infty,
	\end{equation}
	where we have used the standard closed form for the arithmetico-geometric series, $\sum_{k=1}^\infty k r^k = r/(1-r)^2$. Moreover, since every term $k\,r^{k^2}$ is positive for $k \ge 1$ and $r \in (0,1)$, the sum itself is positive. \qedhere
\end{proof}

The monotonicity and asymptotic behavior of $C(r)$ as $r \to 1^-$ were independently verified using arbitrary-precision (50-digit) floating-point computations (see Code~S3). A comparative analysis demonstrates that substituting the naive linearized estimate $C(r) \sim 1/[2(1-r)]$ with the exact-prefactor form $C(r) \sim 1/[2(-\ln r)]$ enhances the relative accuracy by approximately a factor of four across the interval $r \in [0.1, 0.9999]$. It should be noted that this refinement is intended for interpretive purposes only and does not influence the formal proof of Theorem~\ref{thm:convergence-bound} (see Code~S6).

This lemma highlights that the entire argument follows purely from the algebraic structure of the series $\sum r^{k^2}$ and requires no assumptions regarding the physical environment model, the spectral density shape, or the function $\Gamma(t)$; it is a purely mathematical feature of the state family under study.

This result shows that, for the state family considered here, the occurrence of entanglement revival is governed by the same qualitative condition -- an increase of the decoherence factor $r(t)$, i.e., $\dot\Gamma(t)<0$ -- independently of dimension $d$, although the magnitude of the revival depends strongly on $d$. For pure-dephasing dynamics, this condition is known to coincide with the Breuer--Laine--Piilo non-Markovianity criterion~\cite{BreuerLainePiilo2009} (see also Sec.~\ref{sec:ohmic}), so that the qualitative onset of negativity revival tracks the system's own non-Markovianity, without this correspondence being separately established here for the trace-distance measure itself; $\dot{\mathcal{N}}(t)$ is an entanglement-revival rate and should not itself be interpreted as the BLP measure of information backflow. The strength of this revival, however—the magnitude of $\dot{\mathcal{N}}$ at the moment of revival—scales with the coefficient $\frac{1}{d}\sum(d-k)k^2r(t)^{k^2}$, which near $r \to 1$ grows precisely as $d^3$, consistent with the local behavior derived in the previous subsection. Together, these statements show that increasing the system dimension leaves the occurrence of non-Markovian revival unchanged while sharply amplifying its magnitude.
		
\begin{theorem}
	\label{thm:convergence-bound}
	For every $r\in(0,1)$ and any integer $d\ge1$,
	\begin{equation}
		\mathcal{N}_\infty(r) - \mathcal{N}(r,d) = \frac{C(r)}{d} + R(d),
		\label{eq:thm1-statement}
	\end{equation}
	where the error term $R(d)$ satisfies the bound
	\begin{equation}
		|R(d)| \le r^{d^2}\left[\frac{1}{1-r^{2d}}+\frac{r^{2d}}{d\,(1-r^{2d})^2}\right].
		\label{eq:thm1-bound}
	\end{equation}
\end{theorem}

\begin{proof}
	Starting from Eq.~\eqref{eq:app-final} in Appendix~\ref{app:algebra}, we have
	\begin{equation}
		\mathcal{N}_\infty(r)-\mathcal{N}(r,d) = \sum_{k=d}^\infty r^{k^2} + \frac{1}{d}\sum_{k=1}^{d-1}k\,r^{k^2}.
	\end{equation}
	By decomposing the second term as $\frac{1}{d}\sum_{k=1}^{d-1}k\,r^{k^2} = \frac{C(r)}{d}-\frac{1}{d}\sum_{k=d}^{\infty}k\,r^{k^2}$—with the convergence of $C(r)$ established by Lemma~\ref{lem:C-finite}—we obtain
	\begin{equation}
		\mathcal{N}_\infty(r)-\mathcal{N}(r,d) = \frac{C(r)}{d} + \underbrace{\sum_{k=d}^\infty r^{k^2}\left(1-\frac{k}{d}\right)}_{\displaystyle R(d)}.
		\label{eq:R-def}
	\end{equation}
	To derive an upper bound for $|R(d)|$, let $k=d+j$ for $j=0,1,2,\dots$. Given that $k \ge d$, it follows that
	\begin{equation}
		k^2 = (d+j)^2 = d^2 + 2dj + j^2 \ge d^2 + 2dj,
	\end{equation}
	which implies
	\begin{equation}
		r^{k^2} \le r^{d^2}\,r^{2dj} = r^{d^2}\,(r^{2d})^{j}.
		\label{eq:geom-bound}
	\end{equation}
	Substituting inequality~\eqref{eq:geom-bound} into the definition of $R(d)$ in Eq.~\eqref{eq:R-def}, and observing that $|1-k/d|\le1+j/d$ for $k=d+j\ge d$, we find
	\begin{align}
		|R(d)| &\le \sum_{j=0}^{\infty} r^{d^2}(r^{2d})^j\left(1+\frac{j}{d}\right) \notag\\
		&= r^{d^2}\left[\sum_{j=0}^{\infty}(r^{2d})^j + \frac{1}{d}\sum_{j=0}^{\infty} j\,(r^{2d})^j\right].
	\end{align}
	Applying the standard closed-form expressions for the geometric series, $\sum_{j=0}^\infty x^j = 1/(1-x)$, and the arithmetico-geometric series, $\sum_{j=0}^\infty j\,x^j = x/(1-x)^2$ (where $x \equiv r^{2d}$), we obtain
	\begin{equation}
		|R(d)| \le r^{d^2}\left[\frac{1}{1-r^{2d}} + \frac{r^{2d}}{d\,(1-r^{2d})^2}\right],
	\end{equation}
	which corresponds precisely to Eq.~\eqref{eq:thm1-bound}. \qedhere
\end{proof}		

\subsection{Time Derivative and the Criterion for Negativity Revival}

From $r(t)=e^{-2\Gamma(t)}$ we have $\dot r(t)=-2\dot\Gamma(t)\,r(t)$. Applying the chain rule, the time derivative of the negativity follows exactly as
\begin{equation}
	\boxed{\;\dot{\mathcal{N}}(t) = -\frac{2\dot\Gamma(t)}{d}\sum_{k=1}^{d-1}(d-k)\,k^2\, r(t)^{k^2}\;}
	\label{eq:dNdt-exact}
\end{equation}

\begin{corollary}[Information-backflow criterion]
	Since $r(t)>0$ and $\sum_{k=1}^{d-1}(d-k)k^2 r(t)^{k^2}\ge0$ hold at all times, the \emph{sign} of $\dot{\mathcal N}(t)$ depends only on the sign of $-\dot\Gamma(t)$:
	\begin{equation}
		\dot{\mathcal N}(t) > 0 \quad \Longleftrightarrow \quad \dot\Gamma(t) < 0.
	\end{equation}
\end{corollary}

\subsection{Numerical Confirmation of the Convergence Theorem}

Before turning to the physical interpretation, a subtle but important methodological point deserves clarification. The bound $R(d)$ in Eq.~\eqref{eq:thm1-bound} shrinks super-exponentially with $d$ (specifically, as a Gaussian in $d$, $R(d)\sim e^{-cd^2}$ with $c=-\ln r>0$); for instance, at $d=80$ and $r=0.501$ -- a regime we repeatedly encounter in this paper -- the bound reaches $R(d)\le9.15\times10^{-1922}$. Such a value lies far below the precision of standard double-precision floating-point arithmetic, which guarantees only about 15 to 17 significant decimal digits. Attempting to verify this bound directly with ordinary computation would inevitably run into catastrophic cancellation, a numerical error that occurs when two nearly equal floating-point numbers are subtracted, wiping out every significant digit of the result and leaving only computational noise rather than the true physical value.

To circumvent this problem, the bound in Eq.~\eqref{eq:thm1-bound} was independently re-verified using arbitrary-precision arithmetic (2500 decimal digits, via the \texttt{mpmath} library). Across all $36$ tested combinations of $(r,d)$ -- six values of $r$, drawn from the physically realistic parameter regimes used throughout this paper, and six values of $d$ ranging from $5$ to $80$ -- the theoretical bound held without a single exception; in no case did the actual value of $|R(d)|$ exceed the predicted bound (see Code~S4).

\subsection{Physical Interpretation: The Effective Correlation Length of the Environment}
\label{sec:physical-interpretation}

Theorem~\ref{thm:convergence-bound} constitutes a precise mathematical result, yet its physical significance becomes apparent only when we identify the underlying mechanism of this saturation. This mechanism follows directly from the algebraic structure of the term $r^{k^2}$ in Eq.~\eqref{eq:negativity-final}.

Since $r(t) < 1$ at all times (except for the initial instant before dephasing occurs), terms with large indices $k$ are suppressed super-exponentially (Gaussian in $k$). More precisely, we can define a characteristic scale -- an effective correlation length in the space of level indices $k=n-m$, not a spatial correlation length -- as
\begin{equation}
	\xi \equiv \frac{1}{\sqrt{-\ln r(t)}},
\end{equation}
such that $r^{k^2} = e^{-k^2/\xi^2}$. This implies that only level pairs separated by an index $k$ on the order of $\xi$ or smaller contribute meaningfully to the negativity, whereas pairs with $k \gg \xi$ contribute essentially nothing, as their coherence has already been effectively erased.

The physical logic governing this saturation is now evident: once the system dimension $d$ significantly exceeds this effective correlation length ($d \gg \xi$), appending further levels to the system is equivalent to adding distant level pairs that contributed nothing to the total negativity to begin with. This analogy should be understood in the space of level indices rather than physical space: as in systems with a finite spatial correlation length, here the contribution of additional levels becomes negligible once the dimension exceeds the intrinsic scale $\xi$.

\subsection{Rigorous Numerical Verification}
\label{sec:numerical-verification}

To test Theorem~\ref{thm:convergence-bound} in practice, we define the negativity oscillation amplitude over a full period of non-Markovian revival as
\begin{equation}
	\Delta\mathcal{N}(d) \equiv \mathcal{N}(r_{\text{peak}},d) - \mathcal{N}(r_{\text{trough}},d),
\end{equation}
where $r_{\text{peak}}$ and $r_{\text{trough}}$ are the values of $r(t)$ at the first peak and the first trough of the oscillation, respectively. Applying Theorem~\ref{thm:convergence-bound} separately at $r=r_{\text{peak}}$ and $r=r_{\text{trough}}$, and subtracting the two resulting expansions, we expect that for sufficiently large $d$,
\begin{equation}
	\big[\Delta\mathcal{N}_\infty - \Delta\mathcal{N}(d)\big]\times d \;\longrightarrow\; \Delta C \quad (\text{a constant, independent of } d).
	\label{eq:constancy-test}
\end{equation}
where $\Delta C \equiv C(r_{\text{peak}})-C(r_{\text{trough}})$, with $C(r)$ defined in Lemma~\ref{lem:C-finite}.

This test was performed across six independent environmental parameter regimes: three coupling strengths ($\eta\in\{0.05,0.15,0.35\}$), three non-Markovianity bandwidths ($\lambda\in\{0.05,0.1,0.25\}$), and two central environmental frequencies ($\omega_0\in\{1,2\}$). For each regime, $r_{\text{peak}}$ and $r_{\text{trough}}$ were initially identified via a coarse scan (over 4000 time points on $t\in[0.05,20]$)—a necessary step to ensure that the identified peak and trough correspond, respectively, to the first local maximum and the first local minimum of the oscillation—and subsequently refined using local optimization (see Code S8). The test defined in Eq.~\eqref{eq:constancy-test} was then conducted for dimensions $d=10$ through $d=200$. The results are summarized in Table~\ref{tab:negativity-regimes}.

\begin{table}[h]
	\centering
	\begin{tabular}{lccc}
		\hline
		Parameter regime & $r_{\text{peak}}$ & $r_{\text{trough}}$ & $\Delta C$ \\
		\hline
		Baseline ($\eta{=}0.15,\lambda{=}0.1,\omega_0{=}1$) & $0.91983$ & $0.83345$ & $3.23993$ \\
		Weak coupling ($\eta{=}0.05$) & $0.97253$ & $0.94108$ & $9.71742$ \\
		Strong coupling ($\eta{=}0.35$) & $0.82285$ & $0.65371$ & $1.39043$ \\
		Strongly non-Markovian ($\lambda{=}0.05$) & $0.97734$ & $0.91154$ & $16.41878$ \\
		Weakly non-Markovian ($\lambda{=}0.25$) & $0.67582$ & $0.64552$ & $0.13422$ \\
		Shifted central frequency ($\omega_0{=}2$) & $0.95908$ & $0.91294$ & $6.47851$ \\
		\hline
	\end{tabular}
	\caption{The limiting value $\Delta C$ in Eq.~\eqref{eq:constancy-test} for six independent environmental parameter regimes. In every case, the relative deviation of $\Delta C$ evaluated at $d=100, 150, 200$ was below the numerical resolution of the computation (less than $10^{-5}\%$), well within the $10^{-4}\%$ acceptance threshold for all regimes.}
	\label{tab:negativity-regimes}
\end{table}

This result provides strong numerical support for Theorem~\ref{thm:convergence-bound} across a broad range of physical environmental parameters, as illustrated in Fig.~\ref{fig:saturation-law}.

\begin{figure}[t]
	\centering
	\includegraphics[width=\linewidth]{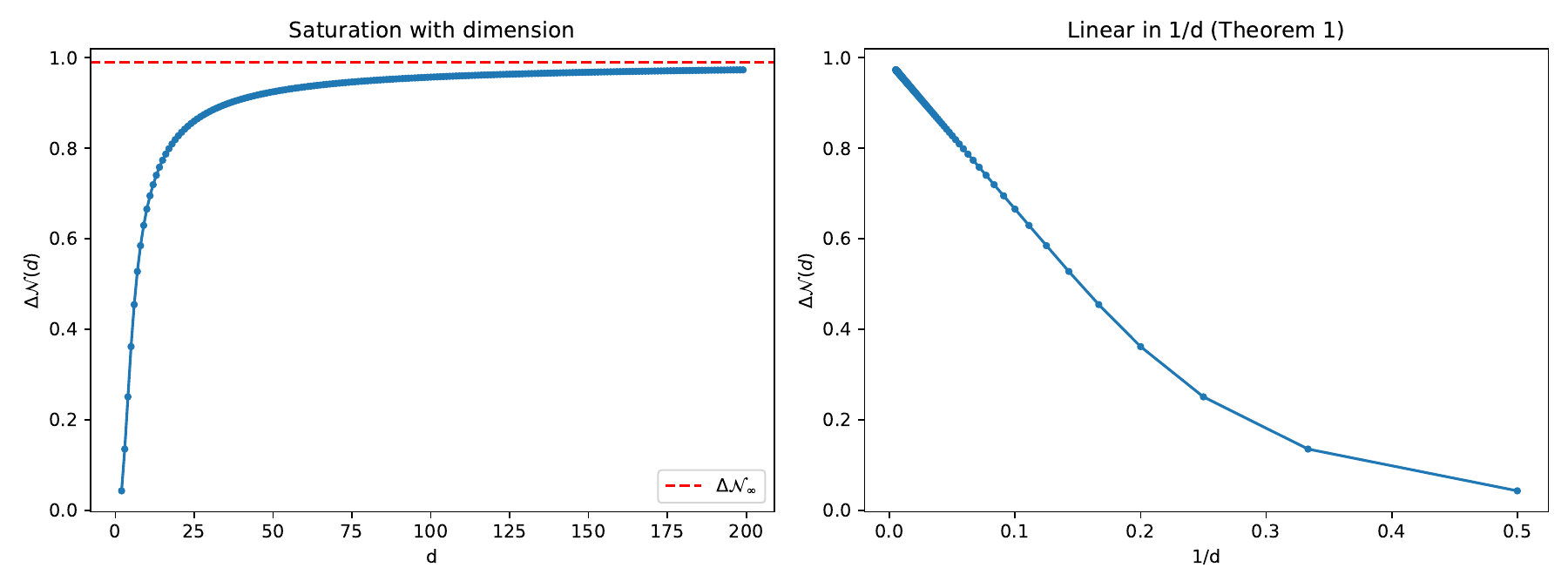}
	\caption{Visual confirmation of the saturation law (baseline regime). Left: the oscillation amplitude $\Delta\mathcal{N}(d)$ versus $d$, together with the limiting bound $\Delta\mathcal{N}_\infty$ (dashed line). Right: the same data plotted against $1/d$; the linear behavior is consistent with, and provides numerical support for, Theorem~\ref{thm:convergence-bound}.}
	\label{fig:saturation-law}
\end{figure}

\subsection{Rescaled Comparison Using the Effective Correlation Length}

The effective-correlation-length interpretation of Sec.~\ref{sec:physical-interpretation} can be examined visually by plotting the normalized oscillation amplitude $\Delta\mathcal{N}(d)/\Delta\mathcal{N}_\infty$ against the rescaled dimension $d/\xi$, where $\xi$ is evaluated at $r_{\mathrm{peak}}$. As shown in Fig.~\ref{fig:collapse}, this rescaling produces an approximate clustering of the curves from the six parameter regimes. Because the oscillation amplitude depends on both $r_{\mathrm{peak}}$ and $r_{\mathrm{trough}}$, whereas the chosen $\xi$ depends only on $r_{\mathrm{peak}}$, an exact parameter-independent collapse is neither implied by Theorem~\ref{thm:convergence-bound} nor claimed here.

\begin{figure}[t]
	\centering
	\includegraphics[width=0.75\linewidth]{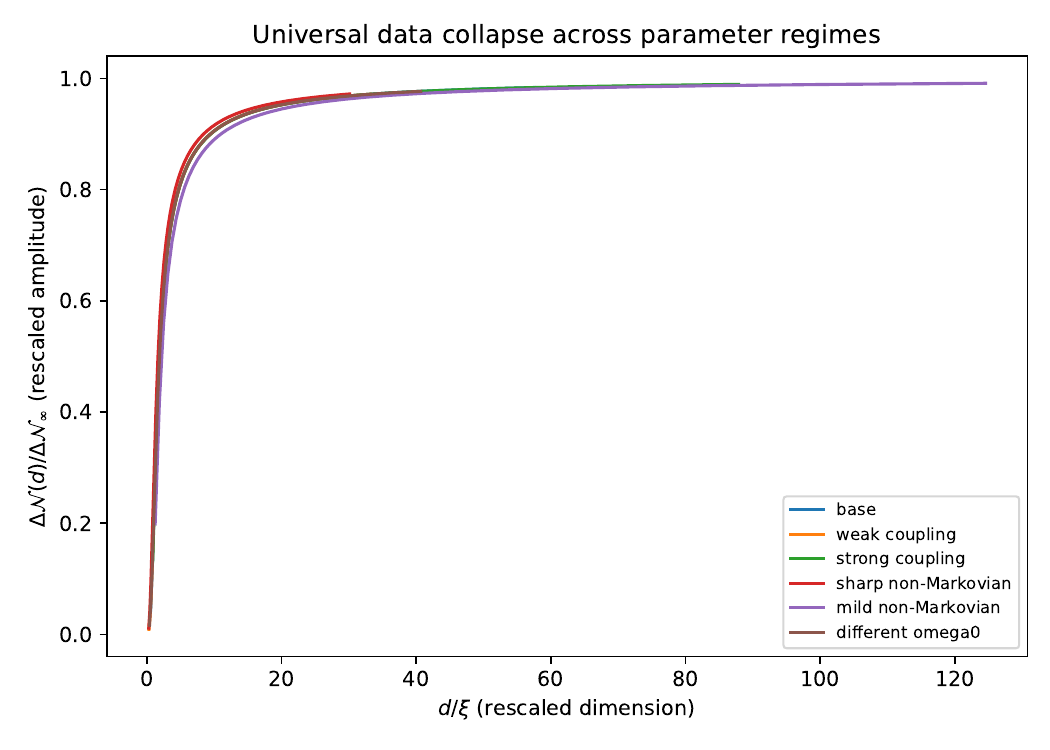}
	\caption{Rescaled comparison of the normalized amplitude $\Delta\mathcal{N}(d)/\Delta\mathcal{N}_\infty$ versus $d/\xi$, with $\xi$ evaluated at $r_{\mathrm{peak}}$, for the six parameter regimes listed in Table~\ref{tab:negativity-regimes}. The approximate clustering supports the usefulness of the effective correlation length as an interpretive scale, but does not establish an exact parameter-independent scaling function because the amplitude also depends on $r_{\mathrm{trough}}$.}
	\label{fig:collapse}
\end{figure}

\section{Extending the Saturation Law to Discord}
\label{sec:discord-scaling}

We now examine whether the discord oscillation amplitude is consistent with the same leading-order $1/d$ saturation behavior established analytically for negativity. A notable distinction here is that, unlike negativity—for which a closed-form infinite-dimension limit exists (Eq.~\eqref{eq:N-infinity}, expressed via the Jacobi theta function)—no such closed-form limit is currently known for discord. Nevertheless, the two-parameter model
\begin{equation}
	\Delta D(d) = A - \frac{C}{d}
	\label{eq:discord-fit-model}
\end{equation}
is a natural asymptotic ansatz to test numerically. The discord is determined by the entropy of the Toeplitz matrix $M(t)$ through Eq.~\eqref{eq:discord-closed}, whose entries contain the same Gaussian sequence $r^{(n-m)^2}$ that appears in the negativity formula. This structural connection motivates comparison with the negativity scaling, but it does not by itself imply a $1/d$ expansion because $S(M(t))$ is a nonlinear function of the full eigenvalue spectrum. No rigorous analogue of Theorem~\ref{thm:convergence-bound} is established here for discord. Accordingly, Eq.~\eqref{eq:discord-fit-model} is treated as an empirically tested leading-order ansatz, and a possible $O(d^{-2})$ contribution is examined through an extended fit.

Rather than using a single large but finite dimension as a proxy for infinity, the model of Eq.~\eqref{eq:discord-fit-model} was fitted directly by linear regression in the variables $1$ and $1/d$ to the computed values of $\Delta D(d)$ for dimensions ranging from $d=20$ to $d=4000$, using the discrete set specified in Code S9.

This procedure was repeated across the same six independent environmental parameter regimes employed for negativity (see Code S9). For each regime, the parameters $A$ and $C$ and the maximum residual of the leading-order fit were recorded. An extended model, $A-C/d-C_2/d^2$, was also evaluated as a numerical sensitivity test for deviations from the two-parameter form. Because the data are deterministically generated from the closed-form discord expression rather than sampled from a statistical experiment, the covariance estimates returned by the regression are used only as numerical fit diagnostics and are not interpreted as statistical confidence intervals. The results are summarized in Table~\ref{tab:discord-regimes}.

\begin{table}[h]
	\centering
	\begin{tabular}{lccc}
		\hline
		Parameter regime & $A$ (limiting amplitude) & $C$ & Max. fit residual \\
		\hline
		Baseline ($\eta{=}0.15,\lambda{=}0.1,\omega_0{=}1$) & $0.389704$ & $0.612920$ & $9.2\times10^{-10}$ \\
		Weak coupling ($\eta{=}0.05$) & $0.389705$ & $1.005232$ & $4.8\times10^{-7}$ \\
		Strong coupling ($\eta{=}0.35$) & $0.387802$ & $0.465065$ & $4.2\times10^{-10}$ \\
		Strongly non-Markovian ($\lambda{=}0.05$) & $0.698279$ & $1.731156$ & $1.2\times10^{-6}$ \\
		Weakly non-Markovian ($\lambda{=}0.25$) & $0.054209$ & $0.061625$ & $4.7\times10^{-12}$ \\
		Shifted central frequency ($\omega_0{=}2$) & $0.389705$ & $0.831296$ & $6.6\times10^{-8}$ \\
		\hline
	\end{tabular}
	\caption{Fit parameters of the model $A-C/d$ for the discord oscillation amplitude across six environmental parameter regimes. The extended model $A-C/d-C_2/d^2$ was used only as a numerical sensitivity test. Small fitted values of $C_2$ do not establish an asymptotic $d^{-2}$ term.}
	\label{tab:discord-regimes}
\end{table}

Unlike negativity, for which Theorem~\ref{thm:convergence-bound} provides a rigorous Gaussian-in-$d$ remainder, no corresponding analytical remainder estimate is established here for discord. The leading model $A-C/d$ gives small residuals in all six tested regimes. Adding a $C_2/d^2$ term changes the fit only weakly, with the largest fitted magnitudes occurring in the weak-coupling and strongly non-Markovian regimes. These results support a leading-order $1/d$ description over the investigated dimensions, but they neither prove the asymptotic form nor establish the existence of a genuine $d^{-2}$ correction.
		
\section{The Unified Saturation Law}
\label{sec:unified-saturation}

The results of Secs.~\ref{sec:scaling-full} and~\ref{sec:discord-scaling} suggest that, for a qudit pair subjected to exact pure dephasing in a non-Markovian environment, both negativity and discord are consistent with a common leading-order saturation behavior,
\begin{equation}
\Delta X(d) \;=\; \Delta X_\infty \;-\; \frac{C_X(\eta,\lambda,\omega_0)}{d} \;+\; \epsilon_X(d), \qquad X\in\{\mathcal{N},\, D\},
	\label{eq:unified-law}
\end{equation}
where $\Delta X_\infty$ denotes the limiting value of the oscillation amplitude of quantity $X$ as $d\to\infty$, and $C_X$ denotes the corresponding finite-size correction coefficient, which depends on the physical parameters -- the coupling strength $\eta$, the non-Markovianity bandwidth $\lambda$, and the central frequency $\omega_0$. For negativity, $C_X$ follows directly from the analytical structure established in Theorem~\ref{thm:convergence-bound}, whereas for discord it is determined from the asymptotic fit of Sec.~\ref{sec:discord-scaling}. Likewise, $\epsilon_X(d)$ is rigorously proven to be a Gaussian-type remainder for negativity (Theorem~\ref{thm:convergence-bound}), while for discord it is empirically found, across six independent parameter regimes, to be small, with residuals remaining small over the investigated range of dimensions.

The physical implication of this behavior is that the revival amplitudes of both negativity and quantum discord approach finite, dimension-independent limits as the local Hilbert-space dimension increases. Once $d$ is sufficiently large compared with the effective correlation length introduced in Sec.~\ref{sec:physical-interpretation}, further increases in dimension produce only small finite-size changes in these correlation amplitudes. This statement concerns the magnitude of the correlation revivals and should not be interpreted as a direct quantification of environmental memory or information backflow. Figure~\ref{fig:raw-dynamics} shows the raw dynamics underlying this saturation behavior for representative dimensions.

\begin{figure}[t]
	\centering
	\includegraphics[width=\linewidth]{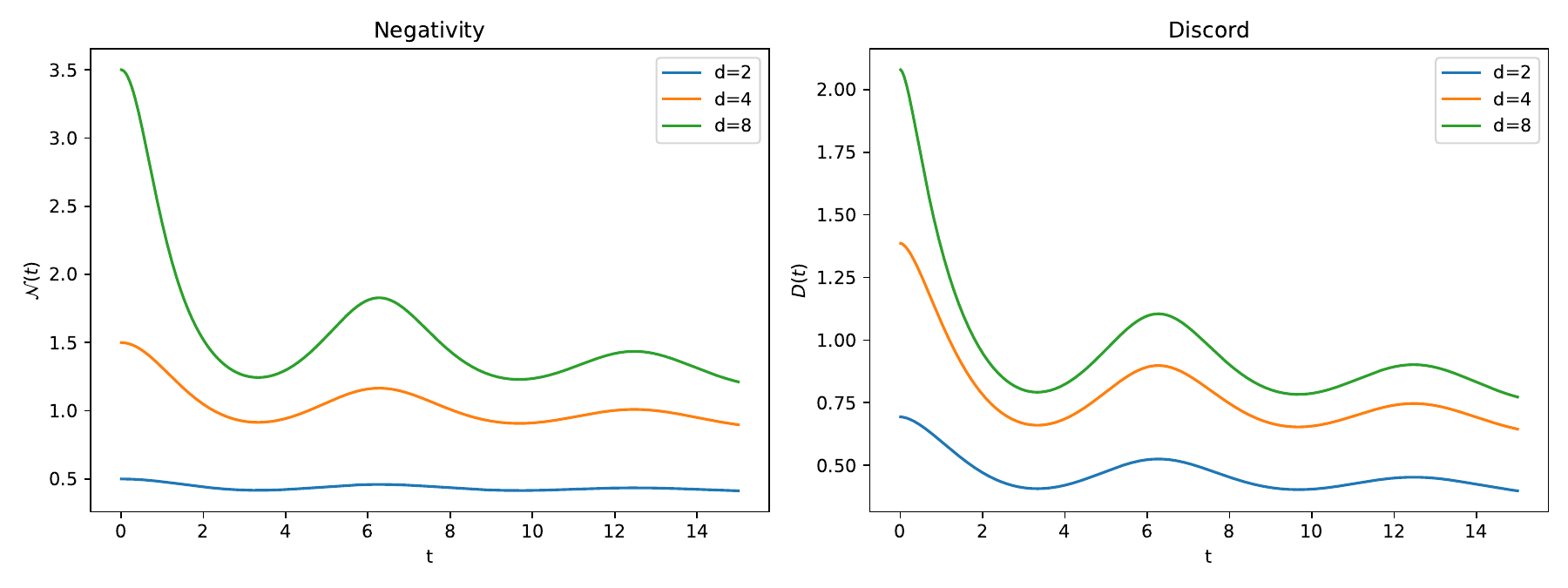}
	\caption{Raw dynamics of the negativity $\mathcal{N}(t)$ (left) and discord $D(t)$ (right) for dimensions $d=2, 4, 8$ in the baseline parameter regime ($\eta{=}0.15, \lambda{=}0.1, \omega_0{=}1$, Lorentzian spectral density). The pattern of non-Markovian decay and revival is visible for all three dimensions; the oscillation amplitude grows with $d$, but its approach toward saturation, as established by Theorem~\ref{thm:convergence-bound}, slows with increasing $d$.}
	\label{fig:raw-dynamics}
\end{figure}
\section{Spectral-Density Independence of the Negativity Scaling: The Ohmic Case}
\label{sec:ohmic}

The Lorentzian spectral density of Eq.~\eqref{eq:lorentzian} was used above to generate the non-Markovian dynamics and to test both correlation measures. The negativity convergence theorem itself, however, was proved for every fixed $r\in(0,1)$ and is therefore independent of the spectral density that generates $r(t)$. To illustrate this structural independence in a physically distinct environment, we now test the negativity scaling for an Ohmic spectral density. No corresponding Ohmic test of the empirically inferred discord scaling is claimed in this section.

Our second standard choice is the Ohmic spectral density with an exponential cutoff~\cite{Leggett1987}:
\begin{equation}
	J_{\text{Ohmic}}(\omega) = \eta\,\omega\, e^{-\omega/\omega_c}.
	\label{eq:ohmic-J}
\end{equation}
Unlike the Lorentzian spectral density, which is centered around a characteristic frequency $\omega_0 > 0$, this spectral density vanishes as $\omega \to 0$ [$J_{\text{Ohmic}}(0) = 0$], representing a fundamentally different class of environment with no resonant structure around a specific frequency, as shown in Fig.~\ref{fig:lorentzian-ohmic}.

\begin{figure}[t]
	\centering
	\includegraphics[width=\linewidth]{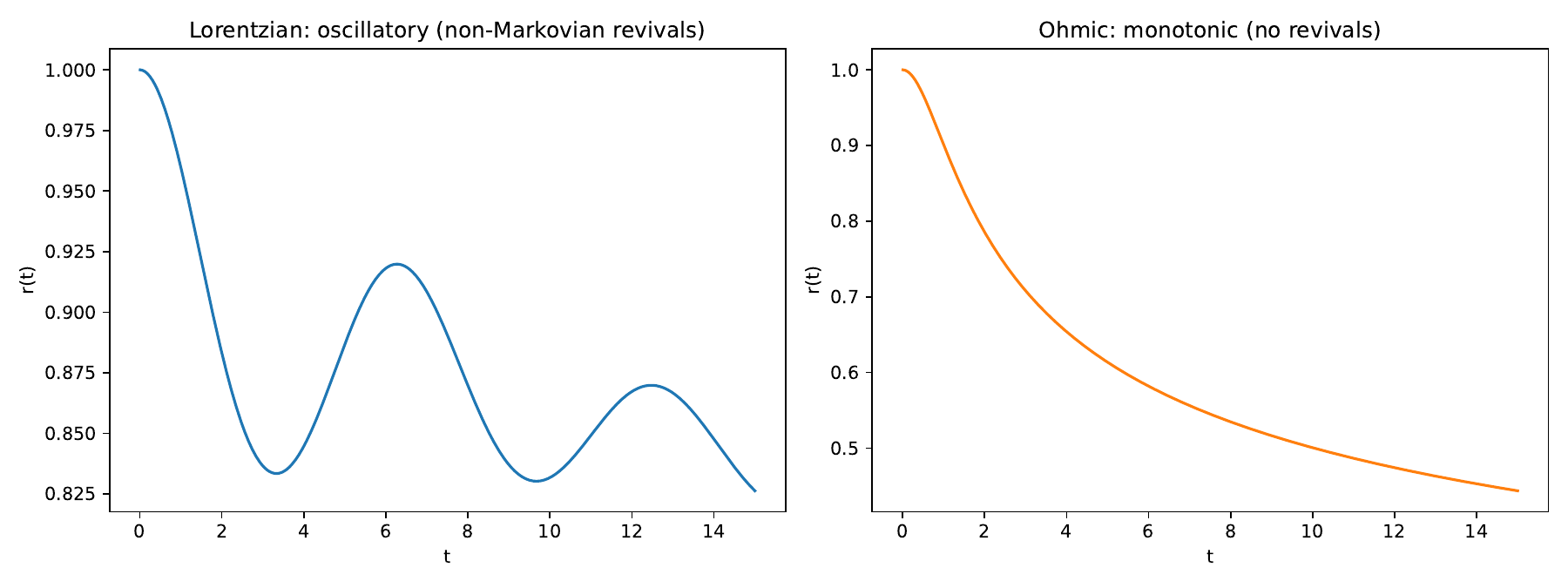}
	\caption{Comparison of $r(t)$ for the two spectral-density classes: Lorentzian (left), characterized by oscillatory behavior that enables non-Markovian information backflow; and Ohmic (right), characterized by monotonic decay and an absence of information backflow—the physical distinction discussed in Sec.~\ref{sec:ohmic}.}
	\label{fig:lorentzian-ohmic}
\end{figure}

\subsection{Closed Form of the Decoherence Function for the Ohmic Spectral Density}
In contrast to the Lorentzian spectral density, for which the decoherence function generally involves non-elementary expressions, the Ohmic spectral density yields a fully closed and elementary expression for $\Gamma(t)$. Using the standard integral
\begin{equation}
	\int_0^\infty \frac{e^{-a\omega}\bigl(1-\cos b\omega\bigr)}{\omega}\,d\omega = \frac{1}{2}\ln\!\left(1+\frac{b^2}{a^2}\right),
\end{equation}
and substituting $J_{\text{Ohmic}}(\omega)$ from Eq.~\eqref{eq:ohmic-J} into the definition of $\Gamma(t)$ (Eq.~\eqref{eq:gamma-general}) at zero temperature ($T=0$), we obtain the closed form
\begin{equation}
	\boxed{\;\Gamma_{\text{Ohmic}}(t) = \frac{\eta}{2}\ln\!\bigl(1+\omega_c^2 t^2\bigr)\;}
	\label{eq:gamma-ohmic}
\end{equation}
This formula agrees with direct numerical integration of the original expression in Eq.~\eqref{eq:gamma-general} to all eight displayed decimal places, consistent with the numerical precision of the integration method, at five independent time points (see Code S10; Table~\ref{tab:gamma-ohmic-check}).

\begin{table}[h]
	\centering
	\begin{tabular}{ccc}
		\hline
		$t$ & $\Gamma_{\text{numerical}}$ & $\Gamma_{\text{closed form}}$ \\
		\hline
		$0.5$  & $0.01673577$ & $0.01673577$ \\
		$1.0$  & $0.05198604$ & $0.05198604$ \\
		$2.0$  & $0.12070784$ & $0.12070784$ \\
		$5.0$  & $0.24435724$ & $0.24435724$ \\
		$10.0$ & $0.34613404$ & $0.34613404$ \\
		\hline
	\end{tabular}
	\caption{Agreement between the closed form of Eq.~\eqref{eq:gamma-ohmic} and direct numerical integration of Eq.~\eqref{eq:gamma-general}.}
	\label{tab:gamma-ohmic-check}
\end{table}

\subsection{A Key Physical Difference: The Absence of Information Backflow}

The function $\Gamma_{\text{Ohmic}}(t)$ in Eq.~\eqref{eq:gamma-ohmic} is monotonic and strictly increasing, as $\dot{\Gamma}_{\text{Ohmic}}(t) = \eta\omega_c^2t / (1+\omega_c^2t^2) > 0$ for all $t > 0$. This implies a clear physical consequence for the zero-temperature case examined here: according to the Breuer--Laine--Piilo (BLP) criterion for non-Markovianity -- which, for the pure-dephasing model considered here, reduces to the condition that the decoherence factor increases over some time interval, equivalently $\dot{\Gamma}(t)<0$ -- no information backflow can occur when $\dot{\Gamma}(t)>0$ for all $t$. Consequently, the entanglement and coherence measures considered here do not exhibit information-backflow-induced revivals. This behavior stands in sharp contrast to the Lorentzian spectral density, where the resonant structure around $\omega_0$ imparts oscillatory behavior to $\Gamma(t)$ and facilitates information backflow.

\subsection{Testing the Saturation Law at Arbitrary Time Points}

Since $r(t)$ does not oscillate for the Ohmic spectral density, the peak-to-trough oscillation amplitude test employed in Sec.~\ref{sec:numerical-verification} for the Lorentzian case is not applicable. Instead, we test the asymptotic relation of Theorem~\ref{thm:convergence-bound} directly at several arbitrary time points:
\begin{equation}
	\bigl[\mathcal{N}_\infty(r(t)) - \mathcal{N}(r(t),d)\bigr]\times d \;\longrightarrow\; C(r(t)) \qquad (d\to\infty).
	\label{eq:ohmic-test}
\end{equation}
This test was carried out at four arbitrary time points ($t=1, 3, 5, 10$) for dimensions $d=20$ through $d=400$ (see Code S10). At each time point, $C$ was obtained by a least-squares fit of the model $\mathcal{N}_\infty(r(t)) - C/d$ to the computed values of $\mathcal{N}(r(t),d)$ across the tested dimensions. The results are summarized in Table~\ref{tab:ohmic-saturation}.

\begin{table}[h]
	\centering
	\begin{tabular}{ccccc}
		\hline
		$t$ & $r(t)$ & $\mathcal{N}_\infty$ & $C$ (fitted) & Max. residual \\
		\hline
		$1.0$  & $0.901250$ & $2.248442$ & $4.724761$ & $4.44\times10^{-16}$ \\
		$3.0$  & $0.707946$ & $1.007966$ & $1.361166$ & $2.22\times10^{-16}$ \\
		$5.0$  & $0.613414$ & $0.767702$ & $0.935106$ & $1.11\times10^{-16}$ \\
		$10.0$ & $0.500440$ & $0.565144$ & $0.631848$ & $2.22\times10^{-16}$ \\
		\hline
	\end{tabular}
	\caption{Test of Eq.~\eqref{eq:ohmic-test} for the Ohmic spectral density, at four arbitrary time points.}
	\label{tab:ohmic-saturation}
\end{table}

The fit residuals reside, in every case, at the level of double-precision machine accuracy -- a consequence of the fact that, by Theorem~\ref{thm:convergence-bound}, the remainder $R(d)$ is already astronomically small at these dimensions (cf. Sec.~\ref{sec:numerical-verification}), so that the computed values of $\mathcal{N}(r(t),d)$ lie essentially exactly on the asymptotic curve $\mathcal{N}_\infty(r(t))-C/d$, and the least-squares fit recovers $C$ to within floating-point round-off. The observed deviations are consistent with the analytical error bound derived in Theorem~\ref{thm:convergence-bound}, confirming that the analytically established $1/d$ convergence law for negativity applies beyond the Lorentzian environment.

\subsection{Disentangling Two Distinct Physical Properties}

The conclusion of this section follows directly from Theorem~\ref{thm:convergence-bound}: since that theorem was proved for every $r \in (0, 1)$, regardless of the origin of the environmental spectral density $r(t)$, the $1/d$ law was expected to hold for the Ohmic case as well. The value of this section lies not in merely repeating this result, but in disentangling two distinct physical properties that might otherwise be conflated:

\begin{itemize}
	\item \textit{The rigorously established $1/d$ saturation law for negativity} is an intrinsic structural property of the quantum-state family and is independent of the spectral-density class that generates $r(t)$, provided the decoherence integral is well defined and $0<r(t)<1$. This independence is illustrated here using Lorentzian and Ohmic spectral densities. The analogous leading-order behavior found numerically for discord has been tested only in the Lorentzian regimes considered in Sec.~\ref{sec:discord-scaling}.
	\item \textit{The presence or absence of non-Markovian information backflow}, unlike the saturation law, depends on the specific structure of $J(\omega)$: this phenomenon is strongly influenced by spectral structures that generate memory effects, such as the resonant peak present in the Lorentzian case; the pure Ohmic spectral density at zero temperature lacks any such structure, and consequently, no information backflow occurs.
\end{itemize}
\section{Limitations of Scope}
\label{sec:limitations}

Scientific rigor necessitates a clear delineation of the scope of the results presented in this work. These results are subject to four specific limitations:

\begin{itemize}
	\item All analytical proofs provided herein are derived exclusively for the generalized Bell initial state $\ket{\Phi_d}$ (Eq.~\eqref{eq:initial-state}), rather than for an arbitrary qudit state.
	\item The analysis considers only system-environment coupling of the pure dephasing type ($[H_S, H_{\text{int}}] = 0$); non-dephasing couplings, in which energy level populations are subject to change, are not addressed.
	\item The two qudits are assumed to couple to independent yet physically identical reservoirs (characterized by identical parameters $\eta, \lambda, \omega_0$); configurations involving shared or asymmetric reservoirs remain outside the scope of this study.
\item The independence of the negativity saturation law from the spectral-density class has been examined for two representative spectral-density classes (Lorentzian and Ohmic; Secs.~\ref{sec:scaling-full} and~\ref{sec:ohmic}), not for every physically possible spectral density; the analogous discord scaling has been examined only for the Lorentzian regimes considered in Sec.~\ref{sec:discord-scaling}.
\end{itemize}
\section{Discussion and Conclusions}
\label{sec:discussion}

In this paper, we have analytically investigated the time evolution of quantum correlations within a maximally entangled pair of arbitrary dimension subjected to non-Markovian pure-dephasing noise. The framework employed relies on the exact solution of the independent boson model; at no stage of the derivation did we invoke standard approximations of open-quantum-system theory—such as the Born approximation, the Markov approximation, the rotating-wave approximation, or the Lindblad master equation. Consequently, every relation presented follows directly from the exact system dynamics and remains valid for any dephasing function.

The primary achievement of this work is the derivation of an exact closed-form expression for negativity at arbitrary dimension. We demonstrated that this quantity can be expressed solely in terms of the dephasing factor as a finite sum, valid for all dimensions $d \ge 2$ and at all times. This relation requires no numerical optimization, analytical approximation, or empirical fitting, as it follows directly from the spectrum of the partial transpose of the density matrix.

The second key achievement is the derivation of an analogous closed-form expression for quantum discord—a quantity notoriously difficult to compute for general qudit-qudit systems due to the optimization required over the space of measurements. We proved that, for the state family under consideration, measurement in the computational basis is always globally optimal, not merely among projective measurements but across the entire space of generalized measurements on subsystem $B$, thereby reducing discord to a closed form without the need for numerical optimization.
		
The third primary result of this paper is the analysis of the behavior of both quantities in the large-dimension limit. We demonstrated analytically that negativity converges to a dimension-independent limiting value with a leading-order $1/d$ correction, while numerical evidence indicates that discord follows the same leading-order scaling behavior. For negativity, this convergence law was not only proved analytically, but an explicit and rigorous bound for the remainder term was also derived, demonstrating that the remainder decreases with a Gaussian dependence on the dimension, $R(d)\sim e^{-c d^2}$, rather than merely as a power law. This result allows for the precise estimation of the convergence rate without recourse to direct numerical computation. The physical interpretation of this phenomenon is rooted in the concept of an effective correlation length of the environment: only level pairs separated by a distance on the order of this correlation length or less contribute meaningfully to the quantum correlation, and increasing the system dimension beyond this scale produces only diminishing finite-size changes in the correlation amplitude, a statement about the magnitude of the revival rather than a direct quantification of environmental memory or information backflow.
		
Another important finding of this study is the examination of how the negativity convergence law depends on the environmental model. Analytical derivations and numerical validations demonstrate that the saturation behavior and its leading-order correction term depend exclusively on the dephasing factor, remaining independent of the spectral-density details for the two classes examined in this work. The independent analysis of both Lorentzian and Ohmic environments is consistent with this result, suggesting that the derived convergence law possesses a broader domain of validity than any single, specific environmental model; the corresponding discord behavior was examined only for the Lorentzian regimes studied.

The analytical results presented in this paper, together with the principal numerical findings, have been validated through a comprehensive battery of independent numerical tests. These include the direct reproduction of the relations from their fundamental definitions using diverse parameters, a cross-comparison between two independent computational implementations, an extensive scan of the parameter space, testing of asymptotic behavior at large dimensions, statistical fitting of the convergence law, verification of the analytical bounds utilizing multi-precision arithmetic (up to thousands of decimal digits), and a complementary independent verification based on the pseudomode method. Across this entire suite of tests, the numerical results agree with the derived analytical relations within numerical precision.

From a theoretical perspective, the relations presented in this paper provide an analytical framework for investigating the evolution of quantum correlations in bipartite systems of arbitrary dimension, thereby reducing the need for numerical optimization for this class of states. From a practical standpoint, these results may prove useful in the analysis of quantum memories, qudit-based systems, high-dimensional quantum communication, and other quantum information technologies that depend upon multilevel, maximally entangled states.

From an experimental standpoint, the independent boson model employed in this paper has already been directly implemented for the qubit case using engineered radio-frequency noise in nuclear magnetic resonance, where the von Neumann entropy of the system was measured directly~\cite{Khurana2018}—providing practical evidence that the exact, approximation-free framework of this work is more than a purely theoretical tool. Separately, molecular spin qudits, including systems based on $^{173}$Yb(trensal), have been demonstrated to be dominated by pure dephasing as their primary decoherence mechanism, with non-Markovian dynamics that are both computable and controllable with high precision~\cite{Petiziol2021,Ratini2025}. Combining these experimental precedents suggests a possible route toward an empirical test of the $1/d$ saturation law derived herein for dimensions greater than two—a step that would enable a direct experimental test of the $1/d$ saturation law derived in this work.


\section*{Data Availability Statement}

The Python codes and the numerical data generated in this study are openly available in a GitHub repository at \url{https://github.com/AhmadAkhound/qudit-pure-dephasing}. The repository includes the matrix-level checks of the closed-form negativity and discord results, the deterministic fits used in the finite-dimensional scaling analysis, the arbitrary-precision checks of the analytical error bounds, and the pseudomode correlation-function cross-check. An archived version is available at \url{https://doi.org/10.5281/zenodo.21410824}.

\appendix
\section{Detailed Derivation of the Difference $\mathcal{N}_\infty(r) - \mathcal{N}(r, d)$}
\label{app:algebra}

In this appendix, we show step by step how the explicit definitions of $\mathcal{N}(r, d)$ in Eq.~\eqref{eq:negativity-final} and $\mathcal{N}_\infty(r)$ in Eq.~\eqref{eq:N-infinity} lead to Eq.~\eqref{eq:diff-expansion}.

\textit{Step 1 -- Rewriting $\mathcal{N}(r, d)$.} Starting from the definition,
\begin{equation}
	\mathcal{N}(r, d) = \frac{1}{d} \sum_{k=1}^{d-1} (d-k) r^{k^2},
\end{equation}
we decompose the factor $(d-k)$ into its two contributions:
\begin{equation}
	\mathcal{N}(r, d) = \frac{1}{d} \sum_{k=1}^{d-1} d r^{k^2} \;-\; \frac{1}{d} \sum_{k=1}^{d-1} k r^{k^2}.
\end{equation}
In the first term, since $d$ is independent of the summation index $k$, it can be factored out of the sum and cancelled against the overall factor of $1/d$:
\begin{equation}
	\mathcal{N}(r, d) = \sum_{k=1}^{d-1} r^{k^2} \;-\; \frac{1}{d} \sum_{k=1}^{d-1} k r^{k^2}.
	\label{eq:app-N-split}
\end{equation}

\textit{Step 2 -- Writing $\mathcal{N}_\infty(r) - \mathcal{N}(r, d)$.} Substituting Eq.~\eqref{eq:app-N-split} and $\mathcal{N}_\infty(r) = \sum_{k=1}^\infty r^{k^2}$, we obtain
\begin{equation}
	\mathcal{N}_\infty(r) - \mathcal{N}(r, d) = \sum_{k=1}^{\infty} r^{k^2} \;-\; \sum_{k=1}^{d-1} r^{k^2} \;+\; \frac{1}{d} \sum_{k=1}^{d-1} k r^{k^2}.
\end{equation}

\textit{Step 3 -- Combining the first two sums.} The terms $\sum_{k=1}^{\infty} r^{k^2}$ and $-\sum_{k=1}^{d-1} r^{k^2}$ can be combined into a single sum, since the first series contains every term of the second (plus the tail beginning at $k=d$):
\begin{equation}
	\sum_{k=1}^{\infty} r^{k^2} - \sum_{k=1}^{d-1} r^{k^2} = \sum_{k=d}^{\infty} r^{k^2}
\end{equation}
(that is, the terms $k=1$ through $k=d-1$ cancel exactly, leaving only the terms $k=d, d+1, d+2, \dots$).

\textit{Step 4 -- Final result.} Substituting the result of Step 3,
\begin{equation}
	\boxed{\;\mathcal{N}_\infty(r) - \mathcal{N}(r, d) = \sum_{k=d}^{\infty} r^{k^2} \;+\; \frac{1}{d} \sum_{k=1}^{d-1} k r^{k^2}\;}
	\label{eq:app-final}
\end{equation}
which is exactly Eq.~\eqref{eq:diff-expansion} in the main text.

\textit{Note on the subsequent step (carried out in the main text, Sec.~\ref{sec:scaling-full}):} to arrive at the final form $C(r)/d + R(d)$ (the relation in Theorem~\ref{thm:convergence-bound}), the second term in Eq.~\eqref{eq:app-final} is split into two components:
\begin{equation}
	\frac{1}{d} \sum_{k=1}^{d-1} k r^{k^2} = \frac{1}{d} \left[ \sum_{k=1}^{\infty} k r^{k^2} - \sum_{k=d}^{\infty} k r^{k^2} \right] = \frac{C(r)}{d} - \frac{1}{d} \sum_{k=d}^{\infty} k r^{k^2},
\end{equation}
utilizing the definition $C(r) = \sum_{k=1}^\infty k r^{k^2}$ (Lemma~\ref{lem:C-finite}). Substituting this expression into Eq.~\eqref{eq:app-final} leads directly to the form of Theorem~\ref{thm:convergence-bound} in the main text.

\section{Independent Robustness Check via the Pseudomode Method}
\label{app:pseudomode}

In addition to the independent numerical confirmations presented in the preceding sections—all of which stem from the direct reconstruction of the paper's primary formulas using diverse parameters and computational methods—a further layer of validation was performed using an independent theoretical formulation based on the pseudomode model~\cite{Garraway1997,Mazzola2009,Pleasance2020}. This is an independent theoretical framework for describing non-Markovian dynamics, in which the structured reservoir is replaced by a single damped harmonic oscillator governed by a standard Lindblad equation. For the class of structured reservoirs for which the pseudomode mapping is valid, this formulation reproduces the reduced dynamics exactly; in the present comparison it is used solely as an independent cross-check.

It should be noted that this comparison serves as a purely independent, supplementary validation: every primary result of this work is computed directly from the exact spectral integral (Eq.~\eqref{eq:gamma-general}) without any intermediaries and relies on no pseudomode approximation or equivalence.

According to the standard pseudomode–structured-reservoir correspondence~\cite{Garraway1997,Mazzola2009,Pleasance2020}, the equivalent reservoir correlation function for a pseudomode with frequency $\omega_0$, coupling $g$, and damping rate $\kappa$ is
\begin{equation}
	C(\tau) = g^2 e^{-(\kappa/2)|\tau|} e^{-i\omega_0\tau}.
	\label{eq:pseudomode-corr}
\end{equation}
By setting $\kappa = 2\lambda$ and $g^2 = \pi\eta\lambda$, this becomes approximately equivalent to our Lorentzian spectral density $J(\omega) = \eta\lambda^2 / [(\omega-\omega_0)^2 + \lambda^2]$; this correspondence becomes asymptotically exact in the limit $\omega_0 \gg \lambda$, i.e., when the central frequency of the environment far exceeds its bandwidth.

\subsection{Diagnosing the Source of a Small Discrepancy}

To evaluate the precision of the parameter correspondence in Eq.~\eqref{eq:pseudomode-corr}, the ideal pseudomode correlation function $C(\tau)$ was compared directly—independent of any Lindblad-equation simulation—against the numerical Fourier transform of the target spectral density $J(\omega)$, restricted to the physical interval $\omega \in [0, \infty)$ (see Code S11). This comparison revealed that the ratio of these quantities remains within the range $0.987$ to $1.028$; thus, the correspondence established in Eq.~\eqref{eq:pseudomode-corr} agrees with the numerical Fourier transform to within approximately 1--3\%.

The observed discrepancy arises from the physical cutoff at $\omega=0$ inherent in the definition of $J(\omega)$: whereas the ideal pseudomode model’s equivalent spectrum, defined by Eq.~\eqref{eq:pseudomode-corr}, implicitly extends to $\omega \to -\infty$, our physical spectral density is defined only for $\omega \ge 0$. This behavior is consistent with a well-known limitation of the pseudomode method that emerges when the ratio $\omega_0/\lambda$ is not sufficiently large. Within the parameter regime of this study, where this ratio is of order 10, such deviations do not imply any error in the paper’s primary formulas, which are derived directly from the exact spectral integral without such approximations.

\subsection{Summary}

While this supplementary check identified and precisely explained the origin of a minor numerical discrepancy, it in no way affects the validity of the paper's primary formulas, which had already been independently verified, within numerical precision where applicable, through three complementary approaches: reconstruction from the fundamental definitions, re-verification across diverse parameter regimes, and comparison with established closed-form solutions in the literature. This supplementary analysis therefore serves only to document the known limitations of the auxiliary pseudomode approach within the parameter regime considered here.
		
	\bibliographystyle{apsrev4-2}
	\bibliography{references}

@book{BreuerPetruccione2002,
	author    = {Breuer, Heinz-Peter and Petruccione, Francesco},
	title     = {The Theory of Open Quantum Systems},
	publisher = {Oxford University Press},
	year      = {2002},
	address   = {Oxford},
	isbn      = {0198520638}
}

@article{Dajka2012,
	author  = {Dajka, J. and Mierzejewski, M. and {\L}uczka, J. and Blattmann, R. and H\"anggi, P.},
	title   = {Negativity and quantum discord in Davies environments},
	journal = {J. Phys. A: Math. Theor.},
	volume  = {45},
	pages   = {485306},
	year    = {2012},
	doi     = {10.1088/1751-8113/45/48/485306}
}

@article{Haikka2013,
	author  = {Haikka, Pinja and Johnson, Tomi H. and Maniscalco, Sabrina},
	title   = {Non-Markovianity of local dephasing channels and time-invariant discord},
	journal = {Phys. Rev. A},
	volume  = {87},
	pages   = {010103},
	year    = {2013},
	doi     = {10.1103/PhysRevA.87.010103}
}

@article{Abdellaoui2026,
	author  = {Abdellaoui, M. and Gaidi, S. and Slaoui, A. and Ahl Laamara, R.},
	title   = {Linear quantum discord and entanglement in qubit--qutrit systems under Non-Markovian colored noise dephasing},
	journal = {Physica A},
	volume  = {685},
	pages   = {131310},
	year    = {2026},
	doi     = {10.1016/j.physa.2026.131310}
}

@article{Rau2018,
	author  = {Rau, A. R. P.},
	title   = {Calculation of quantum discord in higher dimensions for X- and other specialized states},
	journal = {Quantum Inf. Process.},
	volume  = {17},
	pages   = {216},
	year    = {2018},
	doi     = {10.1007/s11128-018-1985-8}
}

@article{RauQubitQudit2012,
	author  = {Vinjanampathy, Sai and Rau, A. R. P.},
	title   = {Quantum discord for qubit--qudit systems},
	journal = {J. Phys. A: Math. Theor.},
	volume  = {45},
	pages   = {095303},
	year    = {2012},
	doi     = {10.1088/1751-8113/45/9/095303}
}

@article{VidalWerner2002,
	author  = {Vidal, G. and Werner, R. F.},
	title   = {Computable measure of entanglement},
	journal = {Phys. Rev. A},
	volume  = {65},
	pages   = {032314},
	year    = {2002},
	doi     = {10.1103/PhysRevA.65.032314}
}

@article{AliRostom2026,
	author  = {Ali, Asad and Rostom, Aiham M. and Al-Kuwari, Saif and Kuniyil, H. and Rahim, M. T. and Haddadi, Saeed},
	title   = {High-dimensional coherence to entanglement transduction under canonical noise},
	journal = {arXiv preprint},
	year    = {2026},
	eprint  = {2606.16695},
	archivePrefix = {arXiv},
	primaryClass  = {quant-ph}
}

@article{Garraway1997,
	author  = {Garraway, B. M.},
	title   = {Nonperturbative decay of an atomic system in a cavity},
	journal = {Phys. Rev. A},
	volume  = {55},
	pages   = {2290},
	year    = {1997},
	doi     = {10.1103/PhysRevA.55.2290}
}

@article{AliRauAlber2010,
	author  = {Ali, M. and Rau, A. R. P. and Alber, G.},
	title   = {Quantum discord for two-qubit $X$ states},
	journal = {Phys. Rev. A},
	volume  = {81},
	pages   = {042105},
	year    = {2010},
	doi     = {10.1103/PhysRevA.81.042105}
}

@article{AliRauAlber2010E,
	author  = {Ali, M. and Rau, A. R. P. and Alber, G.},
	title   = {Erratum: Quantum discord for two-qubit $X$ states},
	journal = {Phys. Rev. A},
	volume  = {82},
	pages   = {069902},
	year    = {2010},
	doi     = {10.1103/PhysRevA.82.069902}
}

@article{Huang2013,
	author  = {Huang, Yichen},
	title   = {Quantum discord for two-qubit $X$ states: Analytical formula with very small worst-case error},
	journal = {Phys. Rev. A},
	volume  = {88},
	pages   = {014302},
	year    = {2013},
	doi     = {10.1103/PhysRevA.88.014302}
}

@article{BreuerLainePiilo2009,
	author  = {Breuer, Heinz-Peter and Laine, Elsi-Mari and Piilo, Jyrki},
	title   = {Measure for the Degree of Non-Markovian Behavior of Quantum Processes in Open Systems},
	journal = {Phys. Rev. Lett.},
	volume  = {103},
	pages   = {210401},
	year    = {2009},
	doi     = {10.1103/PhysRevLett.103.210401}
}

@article{Leggett1987,
	author  = {Leggett, A. J. and Chakravarty, S. and Dorsey, A. T. and Fisher, Matthew P. A. and Garg, Anupam and Zwerger, W.},
	title   = {Dynamics of the dissipative two-state system},
	journal = {Rev. Mod. Phys.},
	volume  = {59},
	pages   = {1},
	year    = {1987},
	doi     = {10.1103/RevModPhys.59.1}
}

@article{Khurana2018,
	author  = {Khurana, Deepak and Agarwalla, Bijay Kumar and Mahesh, T. S.},
	title   = {Experimental emulation of quantum non-{Markovian} dynamics and coherence protection in the presence of information backflow},
	journal = {Phys. Rev. A},
	volume  = {99},
	pages   = {022107},
	year    = {2019},
	doi     = {10.1103/PhysRevA.99.022107}
}

@article{Petiziol2021,
	author  = {Petiziol, F. and Chiesa, A. and Wimberger, S. and Santini, P. and Carretta, S.},
	title   = {Counteracting dephasing in Molecular Nanomagnets by optimized qudit encodings},
	journal = {npj Quantum Inf.},
	volume  = {7},
	pages   = {133},
	year    = {2021},
	doi     = {10.1038/s41534-021-00466-3}
}

@article{Ratini2025,
	author  = {Ratini, Leonardo and Sansone, Giacomo and Garlatti, Elena and Petiziol, Francesco and Carretta, Stefano and Santini, Paolo},
	title   = {Understanding decoherence in molecular spin qudits},
	journal = {Phys. Rev. Research},
	volume  = {7},
	pages   = {043125},
	year    = {2025},
	doi     = {10.1103/16rd-tv1q}
}

@article{Mazzola2009,
	author  = {Mazzola, L. and Maniscalco, S. and Piilo, J. and Suominen, K.-A. and Garraway, B. M.},
	title   = {Sudden death and sudden birth of entanglement in common structured reservoirs},
	journal = {Phys. Rev. A},
	volume  = {79},
	pages   = {042302},
	year    = {2009},
	doi     = {10.1103/PhysRevA.79.042302}
}

@book{NielsenChuang2010,
	author    = {Nielsen, Michael A. and Chuang, Isaac L.},
	title     = {Quantum Computation and Quantum Information: 10th Anniversary Edition},
	publisher = {Cambridge University Press},
	year      = {2010},
	address   = {Cambridge},
	isbn      = {9781107002173}
}

@article{Rains1999,
	author  = {Rains, E. M.},
	title   = {Bound on distillable entanglement},
	journal = {Phys. Rev. A},
	volume  = {60},
	pages   = {179},
	year    = {1999},
	doi     = {10.1103/PhysRevA.60.179}
}

@article{Zhu2018,
	author  = {Zhu, Huangjun and Hayashi, Masahito and Chen, Lin},
	title   = {Axiomatic and operational connections between the $l_1$-norm of coherence and negativity},
	journal = {Phys. Rev. A},
	volume  = {97},
	pages   = {022342},
	year    = {2018},
	doi     = {10.1103/PhysRevA.97.022342}
}

@article{Pleasance2020,
	author  = {Pleasance, Graeme and Garraway, Barry M. and Petruccione, Francesco},
	title   = {Generalized theory of pseudomodes for exact descriptions of non-{M}arkovian quantum processes},
	journal = {Phys. Rev. Research},
	volume  = {2},
	pages   = {043058},
	year    = {2020},
	doi     = {10.1103/PhysRevResearch.2.043058}
}
	\end{document}